\documentclass[11pt,letterpaper]{article}

\usepackage[margin=1in]{geometry}

\usepackage{amsmath}
\usepackage{amssymb}
\usepackage{amsthm}

\usepackage{cite}
\usepackage{xfrac}
\usepackage{varwidth}

\usepackage[utf8]{inputenc}
\usepackage[english]{babel}

\usepackage{array}
\usepackage{lmodern}
\usepackage{times}

\usepackage[vlined,linesnumbered,ruled,algo2e]{algorithm2e}

\usepackage{float}
\usepackage{graphicx}
\usepackage{tikz}

\usepackage{xspace}

\usepackage{subfigure}
\usepackage{mathtools}

\usepackage{etoolbox}
\usepackage{nicefrac}
\usepackage{paralist}

\usepackage{url}
\usepackage{hyperref}
\hypersetup{colorlinks, linkcolor=blue!50!black, citecolor=green!50!black, urlcolor=blue!50!black}
\usepackage{microtype}
\usepackage{bm}
\usepackage[inline]{enumitem}

\usepackage{xcolor}
\def\HiLi{\leavevmode\rlap{\hbox to \hsize{\color{gray!20}\leaders\hrule height .8\baselineskip depth .5ex\hfill}}}

\def\showauthornotes{1}

\ifnum\showauthornotes=1
\newcommand{\Authornote}[2]{{\sffamily\small\color{red}{[#1: #2]}}}
\newcommand{\Authornotecolored}[3]{{\sffamily\small\color{#1}{[#2: #3]}}}
\newcommand{\Authorcomment}[2]{{\sffamily\small\color{gray}{[#1: #2]}}}
\newcommand{\Authorstartcomment}[1]{\sffamily\small\color{gray}[#1: }

\newcommand{\Authorfnote}[2]{\footnote{\color{red}{#1: #2}}}
\newcommand{\Authorfixme}[1]{\Authornote{#1}{\textbf{??}}}
\newcommand{\Authormarginmark}[1]{\marginpar{\textcolor{red}{\fbox{\Large #1:!}}}}
\else
\newcommand{\Authornote}[2]{}
\newcommand{\Authornotecolored}[3]{}
\newcommand{\Authorcomment}[2]{}
\newcommand{\Authorstartcomment}[1]{}

\newcommand{\Authorfnote}[2]{}
\newcommand{\Authorfixme}[1]{}
\newcommand{\Authormarginmark}[1]{}
\fi

\DeclareMathOperator{\spn}{span}
\DeclareMathOperator\rank{rank}
\DeclareMathOperator{\greedy}{Greedy}
\DeclareMathOperator{\generalizedGreedy}{GeneralizedGreedy}
\DeclareMathOperator{\role}{role}

\DeclareMathOperator{\imp}{Rel}
\DeclareMathOperator{\simp}{Simulated}
\DeclareMathOperator{\simpExt}{SimulExt}
\DeclareMathOperator{\online}{Online}
\DeclareMathOperator{\onlineExt}{OnlineExt}
\DeclareMathOperator{\nonzero}{non-zero}
\DeclareMathOperator{\argmax}{argmax}

\def\E{\mathbb{E}}

\newcommand{\MSP}{MSP\xspace}
\newcommand{\MISP}{MISP\xspace}
\newcommand{\cA}{{\mathcal{A}}}
\newcommand{\cE}{{\mathcal{E}}}
\newcommand{\cI}{{\mathcal{I}}}
\newcommand{\cL}{{\mathcal{L}}}
\newcommand{\wh}{\ensuremath{w'}}
\newcommand{\OPTh}{\ensuremath{\OPT'}}
\newcommand{\OPT}{{\mathrm{OPT}}}
\newcommand{\optc}{{\mathrm{opt}}}
\newcommand{\ALG}{\mathcal{A}}

\newtheorem{theorem}{Theorem}[section]
\newtheorem{lemma}[theorem]{Lemma}

\newtheorem{corollary}[theorem]{Corollary}
\newtheorem{definition}[theorem]{Definition}
\newtheorem{proposition}[theorem]{Proposition}

\newtheorem{observation}[theorem]{Observation}

\newtheorem*{claim*}{Claim}

\makeatletter
\newtheorem*{rep@theorem}{\rep@title}
\newcommand{\newreptheorem}[2]{%
\newenvironment{rep#1}[1]{%
 \def\rep@title{#2 \ref{##1}}%
 \begin{rep@theorem}}%
 {\end{rep@theorem}}}
\makeatother

\newreptheorem{theorem}{Theorem}

\title{A Framework for the Secretary Problem on the Intersection of Matroids}

\author{
Moran Feldman\thanks{Dept. of Mathematics and Computer Science, Open University of Israel.
Email:
\href{mailto:moranfe@openu.ac.il}{moranfe@openu.ac.il}.
Supported by Israel Science Foundation grant 1357/16.}
\and
Ola Svensson\thanks{School of Computer and Communication Sciences, EPFL.
Email:
\href{mailto:ola.svensson@epfl.ch}{ola.svensson@epfl.ch}.
Supported by ERC Starting Grant 335288-OptApprox.}
\and
Rico Zenklusen\thanks{
Department of Mathematics, ETH Zurich, Zurich, Switzerland.
Email: \href{mailto:ricoz@math.ethz.ch}%
{ricoz@math.ethz.ch}.
Supported by the Swiss National Science Foundation grant
200021\_165866, ``New Approaches to Constrained
Submodular Maximization''.}
}

\date{\today}

\begin{document}
\maketitle 
\thispagestyle{empty}
\addtocounter{page}{-1}

\begin{abstract}
The secretary problem became one of the most prominent online selection problems due to its numerous applications in online mechanism design. The task is to select a maximum weight subset of elements subject to given constraints, where elements arrive one-by-one in random order, revealing a weight upon arrival. The decision whether to select an element has to be taken immediately after its arrival.
The different applications that map to the secretary problem ask for different constraint families to be handled. The most prominent ones are matroid constraints, which both capture many relevant settings and admit strongly competitive secretary algorithms.
However, dealing with more involved constraints proved to be much more difficult, and strong algorithms are known only for a few specific settings. 
In this paper, we present a general framework for dealing with the secretary problem over the intersection of several matroids. This framework allows us to combine and exploit the large set of matroid secretary algorithms known in the literature. As one consequence, we get constant-competitive secretary algorithms over the intersection of any constant number of matroids whose corresponding (single-)matroid secretary problems are currently known to have a constant-competitive algorithm. Moreover, we show that our results extend to submodular objectives.

\end{abstract}


\newpage

\section{Introduction}

The secretary problem is a classical online selection problem with unclear origins, but dating back to at least the 1950's (see~\cite{ferguson_1989_who,gardner_1960_mathematical,gardner_1960b_mathematical,lindley_1961_dynamic} for more details). In its original form the task is to hire the best out of a known number $n$ of candidates ranked according to an unknown total order. The candidates are interviewed in uniformly random order, and it is only possible to compare pairs of candidates that have already been interviewed. Between any two interviews, one has to decide immediately and irrevocably whether to hire the candidate interviewed last, and the game stops as soon as some candidate gets hired. The goal is to find a hiring/selection strategy that maximizes the probability to hire the best candidate.
A well-known and asymptotically optimal algorithm by Dynkin~\cite{dynkin_1963_optimum} selects the best candidate with probability $\sfrac{1}{e}$.

More recently, there has been a revival of variants of the secretary problem due to a multitude of applications in online mechanism design,
where items are sold to agents arriving online
(see~\cite{%
babaioff_2007_matroids,%
babaioff_2008_online,%
chawla_2010_multi-parameter,%
hajiaghayi_2004_adaptive,%
kleinberg_2005_multiple-choice%
} and references therein). Moreover, links between the secretary problem and further problems in mechanism design, like prophet inequalities and sequential posted pricings, led to a diverse community of researchers interested in the problem.

Most of the recently studied variants of the secretary problem are of the following type. Elements (agents) $e$ of a set $N$ appear one at a time, each revealing a nonnegative weight $w(e)$ when appearing. As before, one has to decide immediately and irrevocably after the appearance of an element whether to select it. The goal is to maximize the expected weight of the selected elements subject to a given set of constraints that the selected elements have to fulfill.
A variety of constraint types have been considered, starting with a simple cardinality constraint imposing that any $k$ elements can be selected~\cite{kleinberg_2005_multiple-choice}.
In particular, matroid constraints%
\footnote{We recall that a matroid is a tuple $M=(N,\mathcal{I})$, where $N$ is a finite set, called \emph{ground set}, and $\mathcal{I}\subseteq 2^N$ is a nonempty family of subsets of $N$, called \emph{independent sets}, fulfilling the following properties: \begin{enumerate*}[label=(\roman*)]
\item if $I\in \mathcal{I}$ and $J\subseteq I$, then $J\in \mathcal{I}$;
\item if $I,J\in \mathcal{I}$ and $|I| > |J|$, then $\exists e\in I\setminus J$ such that $J\cup \{e\}\in \mathcal{I}$.
We refer to~\cite{schrijver_2003_combinatorial} for more information on matroids.
\end{enumerate*}}
have proven to be a very interesting and promising constraint type, capturing a wide set of relevant problems, while offering at the same time many structural properties enabling the design of strong online algorithms. Consequently, matroid constraints have been by far the best-studied constraint class for the secretary problem.

The secretary problem under matroid constraints was introduced by Babaioff et al.~\cite{babaioff_2007_matroids}, and is known as the \emph{Matroid Secretary Problem (MSP)}. Constant-competitive algorithms for MSP exist for numerous classes of matroids, including graphic matroids~\cite{babaioff_2007_matroids,korula_2009_algorithms}, co-graphic matroids~\cite{soto_2013_matroid}, laminar matroids~\cite{im_2011_secretary,jaillet_2013_advances,ma_2013_simulatedSTACS}, transversal matroids~\cite{babaioff_2007_matroids,dimitrov_2012_competitive,korula_2009_algorithms,kesselheim_2013_optimal}, regular and decomposable matroids~\cite{dinitz_2013_matroid}, and linear matroids with a matrix representation having at most a constant number of non-zero entries per column~\cite{soto_2013_matroid}.
Moreover, the above list can be extended by a result of~\cite{babaioff_2007_matroids}, showing that any $c$-competitive algorithm for MSP over a given matroid $M$ can be transformed into a $\Omega(c)$-competitive algorithm for MSP over any truncation of $M$.
For general matroids, $\Omega(1/\log\log\rank)$ is the currently best known competitive ratio, where $\rank$ is the rank of the underlying matroid~\cite{feldman_2015_simple,lachish_2014_competitive}.
It remains open, whether a constant-competitive algorithm for MSP exists, with the well-known Matroid Secretary Conjecture claiming this being the case.

Motivated by various applications, secretary problems subject to constraint families beyond matroids have recently gained attention, including matching constraints in graphs and hypergraphs~\cite{dimitrov_2012_competitive,kesselheim_2013_optimal,korula_2009_algorithms}, knapsack constraints~\cite{babaioff_2007_knapsack}, independent sets in graphs~\cite{gobel_2014_online}, and the intersection of (a constant number of) laminar matroids~\cite{ma_2013_simulatedSTACS}.
Moreover, for arbitrary downward-closed set systems, Rubinstein~\cite{rubinstein_2016_beyond} presents an algorithm for the secretary problem that nearly matches a known information-theoretic hardness bound of $O(\sfrac{\log\log n}{\log n})$ by Babaioff et al.~\cite{babaioff_2007_matroids}, where $n$ is the number of elements.

The flurry of results on MSP and constraints beyond matroids has led to numerous algorithmic ideas.
However, there are still large gaps in our understanding of secretary problems, and despite the many known results for various specific constraint families, generally applicable techniques to design new strong secretary algorithms are still largely missing.
The goal of this work is to partially fill this gap by presenting a general framework for combining MSP algorithms for specific matroids into a strong secretary algorithm over the intersection of these matroids (we use the name \emph{Matroid Intersection Secretary Problem (\MISP)} for the secretary problem over such intersections).
This framework allows us to leverage the extensive existing literature on MSP towards the design of secretary algorithms for a considerably more general class of constraints in an essentially black-box way.
Previously, it was unclear how to obtain nearly-optimal competitive ratios when dealing with combined constraints, and a constant-competitive algorithm for the intersection of a constant number of matroids was only known when the involved matroids were laminar~\cite{ma_2013_simulatedSTACS}.
However, this procedure was heavily specialized to laminar matroids, and hard to extend. In particular, a question in~\cite{im_2011_secretary}, which was open since, asks whether it is possible to obtain an $\Omega(c)$-competitive algorithm for the secretary problem on the intersection of a matroid $M$ and a single laminar matroid, given a $c$-competitive algorithm for MSP on $M$.
As a very special case of our framework, we resolve this question affirmatively (under a mild assumption on the algorithm for $M$). More generally, we get constant-competitive secretary algorithms for the intersection of any constant number of matroids as long as for each involved matroid a constant-competitive MSP algorithm is known (again, assuming this algorithm obeys the above-mentioned mild assumption). Moreover, by extending a technique of~\cite{feldman_2015_submodular}, we show that our results extend to the submodular secretary problem.\footnote{In the submodular secretary problem, the objective function is a nonnegative submodular function $f\colon 2^N \rightarrow \mathbb{R}_{\geq 0}$ on the ground set $N$, i.e., $f$ fulfills $f(A) + f(B) \geq f(A\cup B) + f(A\cap B)$ for any $A,B\subseteq N$. Moreover, access to $f$ is provided via a value oracle, which can be queried with any subset $S\subseteq N$ of elements that revealed so far, and returns the value $f(S)$.}

\subsection{Overview of our results and techniques}

For better clarity, and to highlight our main ideas, we build our framework in several steps based on a sequence of results which we think to be of independent interest. In this section, we give a brief overview of our main techniques and their implications.

In order for our techniques to work, our framework requires that the single-matroid MSP algorithms used within it are \emph{order-oblivious}, a notion introduced in~\cite{azar_2014_prophet}.
An order-oblivious secretary algorithm works in two phases. In the first phase it observes a certain number of all elements without selecting any of them, to learn about their weights. Only in the second phase elements get selected. However, the second phase must not depend on the elements arriving in random order, and has to work even if the elements appear in adversarial order. Order-obliviousness is a quite weak assumption in the sense that most known MSP algorithms are order-oblivious. Moreover, for all matroid classes known to the authors that have been studied in the context of MSP, the currently best competitive ratio is either already achieved by an order-oblivious algorithm, or an order-oblivious algorithm with a competitive ratio at most a constant-factor worse than the best-known one can easily be obtained from existing procedures.
%
Thus, the restriction to order-oblivious algorithms is a rather benign assumption. Furthermore, when combining algorithms it is clear that some assumption of this nature is required (otherwise, the combining technique must be sophisticated enough to combine, for example, two algorithms such that the first of them only selects a subset of the first $n/2$ elements, while the second only selects a subset of the last $n/2$ elements).   
In Section~\ref{sec:preliminaries}, we provide a more formal definition of order-oblivious algorithms.

Let us now present our results for {\MSP} algorithms that are $c$-competitive (for some $c>0$) because they select each element of the offline optimum $\OPT$ with probability at least $c$. We call such algorithms \emph{$\optc$-competitive}, or more precisely \emph{$c$-$\optc$-competitive} (Section~\ref{sec:optSelectToMatInt} contains a formal definition).%
\footnote{Throughout this paper we assume that no two elements have the same weight, which can easily be achieved by breaking ties between elements of identical weight in an arbitrarily (but consistent) way. One important consequence of this assumption is that the offline optimum is unique.}
Many MSP algorithms in the literature are $\optc$-competitive, including algorithms for partition, laminar and co-graphic matroids; however, most are not. At a later stage we significantly weaken the condition of $\optc$-competitiveness to capture the best-known algorithms (up to a constant factor) for all matroid classes that were studied in the context of MSP, to the best of the authors' knowledge.



Consider now the secretary problem over the intersection of two matroids $M_1=(N,\mathcal{I}_1)$ and $M_2=(N,\mathcal{I}_2)$, and assume that we have an order-oblivious $\Omega(1)$-$\optc$-competitive algorithm $\ALG_j$ for MSP over $M_j$, where $j\in \{1,2\}$.
A natural na\"{\i}ve approach would be to run for both $M_1$ and $M_2$ the algorithms $\ALG_1$ and $\ALG_2$ independently in parallel, with the rule that we only select elements that get selected by both procedures simultaneously. However, this approach could have an arbitrarily bad competitive ratio due to the following.
Clearly, since both $\ALG_1$ and $\ALG_2$ are $\Omega(1)$-$\optc$-competitive algorithms, we will select elements of $\OPT_1\cap \OPT_2$ with constant probability, where $\OPT_j$ for $j\in \{1,2\}$ is the offline optimum with respect to only $M_j$.
The issue is that $\OPT_1$ and $\OPT_2$ may have very little or no overlap
(another technicality is to make sure that we do not have bad correlations between the algorithms).
%
A key part of our framework is to show that this issue can be resolved in a general way, such that we can assume to have an instance where the weight of $\OPT_1\cap \OPT_2$ is a constant fraction of the weight of $\OPT$.

More precisely, consider the intersection of $k$ matroids $M_1,\ldots, M_k$ on the same ground set $N$. We show that the following preprocessing procedure leads to a new instance where the $k$ single-matroid offline optima $\OPT_j$, one for each $M_j$ where $j\in[k]\coloneqq\{1,\ldots,k\}$, have large overlap. We start by observing a well-chosen fraction $p$ of all elements, obtaining a set of observed elements $S$. We then consider the solution $G=\greedy(S)\subseteq S$ that the greedy algorithm for the intersection of the $k$ matroids $M_1,\ldots, M_k$ returns when applied to $S$. Finally we define a new MISP instance that, among the non-sampled elements $N\setminus S$, only considers elements $e\in N\setminus S$ such that $\greedy(S\cup \{e\})\neq \greedy(S)$; 
we call such elements \emph{greedy-relevant}, or more precisely, \emph{greedy-relevant with respect to $S$}.\footnote{There are several technical aspects one has to take care of which we did not address in this brief outline. In particular, instead of defining a new MISP instance on a subset of the elements, we really define a preprocessing procedure leading to a new MISP instance on all elements $N$ (but with a modified weight function $w'$). This allows us to use existing MSP algorithms in a black-box fashion.}
Our main technical contribution here is to show that the intersection of the $k$ single-matroid offline optima on only the greedy-relevant elements has a total weight of $\Omega(\frac{1}{k^2}w(\OPT))$.
Furthermore, we show that this allows us to obtain the following.

\begin{theorem}\label{thm:mainSimple}
   Consider $k$ matroids $M_1,\dots, M_k$ on a common ground set $N$. Suppose that, for $j=1,\dots, k$,
  there is an order-oblivious $c_j$-$\optc$-competitive \MSP algorithm
  $\mathcal{A}_j$ on matroid $M_j$. Then, there is an order-oblivious $(\frac{1}{4k^2}\prod_{j=1}^k c_j)$-competitive online
  algorithm for {\MISP} on the intersection $M_1 \cap M_2 \cap
  \dots \cap M_k$.
\end{theorem}

The above theorem already has interesting implications. In particular, constant-competitive and order-oblivious algorithms are known---or can be easily obtained from known algorithms by a slight modification---for several matroids, including partition, laminar and co-graphic matroids.
Hence, Theorem~\ref{thm:mainSimple} implies that the secretary problem on the intersection of any constant number of such matroids admits a constant-competitive algorithm.

We later show how the requirement of $c$-$\optc$-competitiveness can be substantially weakened, by extending our framework to capture, to the best of the authors' knowledge, all classes of matroids that have been studied in the literature in the context of MSP, and leading to the following general result on MISP.
For simplicity, we state the result for the intersection of a \emph{constant} number of matroids; however, in a similar spirit as the statement of Theorem~\ref{thm:mainSimple}, the result can be generalized to a super-constant number of matroids (see Section~\ref{sec:genToMatInt} for more details).

\begin{theorem}\label{thm:mainGeneral}
Consider MISP over the intersection of $\alpha+\beta=O(1)$ matroids, where each of the first $\alpha$ matroids is either a partition, laminar, graphic, co-graphic, regular, max-flow min-cut, column-sparse linear, or transversal matroid---which are the matroids for which $\Omega(1)$-competitive MSP algorithms are known---and each of the remaining $\beta$ matroids can be arbitrary. Then, there is a $\Omega(1/(\log\log\rank)^\beta)$-competitive algorithm for this MISP, where $\rank$ denotes the cardinality of a maximum cardinality common independent set of all $\alpha+\beta$ matroids.%
\footnote{A direct use of the $\Omega(1/\log\log\rank)$-competitive algorithms for MSP on general matroids in our framework would lead to a competitive ratio of $\Omega(1/\prod_{j=1}^\beta \log\log\rank_j)$, where $\rank_j$ is the rank of the $j$-th general matroid. However, in the {\MISP} setting, one can improve this ratio to $\Omega(1/(\log\log\rank)^{\beta})$. This can be achieved by truncating each general matroid such that only $\rank$-many elements can be chosen from it. If the matroids are known upfront, this truncation can be performed upfront; otherwise, $\rank$ has to be first estimated through sampling (see~\cite{feldman_2015_simple} for an application of this technique).
}
\end{theorem}

Theorem~\ref{thm:mainGeneral} allows for leveraging the best-known competitive ratios for the thoroughly studied MSP in the context of MISP. Moreover, it is important to mention that we assume in this theorem that a matroid of a given class is known upfront to the {\MISP} algorithm exactly when the current {\MSP} state of the art algorithm for this class needs this information. Currently these are the algorithms for graphic, co-graphic, regular, max-flow min-cut, column-sparse linear and transversal matroids. 

Using an extension of a technique presented in~\cite{feldman_2015_submodular}, we can also extend secretary algorithms for {\MISP} to deal with nonnegative submodular objectives.

\begin{theorem}\label{thm:mainSubm}
Let $\mathcal{A}$ be an $\alpha$-competitive algorithm for an {\MISP} instance on the intersection of $k$ matroids. Then, this algorithm can be transformed into an $\sfrac{\alpha^2}{128 k^2}$-competitive secretary algorithm for maximizing a nonnegative submodular function $f$ over the intersection of the same matroids.
Moreover, if $\mathcal{A}$ is order-oblivious, then so is the transformed algorithm for submodular {\MISP}.
\end{theorem}


When algorithm $\cA$ has one of a few common properties, the result obtained by~\cite{feldman_2015_submodular} can be improved. This is true also for the result given by Theorem~\ref{thm:mainSubm}. However, to avoid repeating large sections from~\cite{feldman_2015_submodular}, we give here only the basic result, and refer the reader to~\cite{feldman_2015_submodular} for more information about these possible improvements.

\subsection{Further related work}

We highlight that for MISP on the intersection of $k$ matroids it is not hard to obtain an $\Omega(1/(k\cdot \log\rank))$-competitive algorithm by bucketing the weights of the elements in $\Theta(\log\rank)$ many classes and greedily selecting only elements from one randomly chosen class.\footnote{See, for example,~\cite{feldman_2015_simple} for details about how to do the bucketing without knowing the maximum weight upfront.}
However, the thus obtained competitive ratio is typically far from optimal.
Interestingly, Bateni et al.~\cite{bateni_2013_submodular} obtained a generalization of this result achieving a competitive ratio of $\Omega(1/(k\cdot \log^2\rank))$ for {\MISP} over the intersection of $k$ matroids even in the presence of a submodular objective function.
In contrast, our framework implies an $[\Omega(1/\log\log\rank)]^{2k}$-competitive algorithm for the same setting.

Many variants of the secretary problem have been considered in the literature. In particular, different assumptions can be made on the order in which elements arrive and how weights are assigned to the elements. We recall that the secretary problem assumes a uniformly random arrival order and adversarial weights. Progress has been achieved under various other assumptions~\cite{soto_2013_matroid,oveisgharan_2013_variants,jaillet_2013_advances,kesselheim_2015_secretary}.
Increased interest also arose in the secretary problem with nonlinear objectives, with a focus on the maximization of submodular functions~\cite{barman_2012_secretary,bateni_2013_submodular,%
feldman_2011_improved,gupta_2010_constrained,ma_2013_simulatedSTACS,feldman_2015_submodular}.
Moreover, a class of problems that are closely related to the secretary problem are prophet inequalities (see~\cite{kleinberg_2012_matroid,azar_2014_prophet,dutting_2015_polymatroid,feldman_2016_online} and references therein). For this setting, the authors recently introduced a framework, based on an online version of contention resolution schemes (see~\cite{chekuri_2014_submodular}), that allows for combining constraints~\cite{feldman_2016_online}. However, despite the affinity between secretary problems and prophet inequalities, it is not clear how these techniques could be carried over to the secretary problem.

Finally, we want to highlight a nice survey by Dinitz~\cite{dinitz_2013_recent} on {\MSP}, which contains many further links to related results.

\subsection{Organization of the paper}
We start with some preliminaries in Section~\ref{sec:preliminaries}.
In Section~\ref{sec:makeOptsOverlap}, we prove the key result that allows us to reduce a general {\MISP} instance on the intersection of $k$ matroids to an instance where the intersection of the $k$ offline optima for the individual matroids contains an $\Omega(\sfrac{1}{k^2})$-fraction of the weight of an optimum offline solution for the original instance.
Section~\ref{sec:optSelectToMatInt} demonstrates how this result can be used to obtain strong secretary algorithms for {\MISP} when we are given order-oblivious $\optc$-competitive algorithms for {\MSP} with respect to each single matroid, thus obtaining Theorem~\ref{thm:mainSimple}. In Section~\ref{sec:genToMatInt}, we present a generalization of our framework that significantly weakens the requirement of the existence of an $\optc$-competitive {\MSP} algorithm for each single matroid, leading to a broadly applicable framework for {\MISP} based on algorithms for {\MSP}. Theorem~\ref{thm:mainGeneral} is one consequence of this framework.
Appendix~\ref{sec:algForSpecMats} shows how existing {\MSP} algorithms can easily be adapted to our framework.
Finally, Appendix~\ref{sec:submodular} discusses submodular objective functions and proves Theorem~\ref{thm:mainSubm}.

\section{Preliminaries}\label{sec:preliminaries}

As highlighted in the introduction, we focus on order-oblivious algorithms in this paper, which is a benign assumption. Order-oblivious algorithms do not require the (full) random arrival assumption of the secretary problem, but only need to be able to observe a uniformly random subset of elements of a chosen cardinality, before starting to select any elements. To formalize this, we thus assume to work in the following order-oblivious secretary model.

\begin{definition}[Order-oblivious secretary model]
In the order-oblivious secretary model there are $n$ elements $N$ with unknown non-negative weights, and a down-closed constraint family $\mathcal{F}\subseteq 2^N$ given by an independence oracle, i.e., for any $S\subseteq N$ one can check whether $S\in \mathcal{F}$.
An algorithm in this model first specifies a number $m\leq n$ and requests a uniformly random sample of $m$ elements of $N$ (thereby, learning the weights of the sampled elements). This is called the \emph{sampling phase}, and none of the elements observed during it can be selected. The algorithm then observes one-by-one (in adversarial order%
\footnote{Even though this is typically not of central importance, we highlight that the adversarial order may depend on the elements observed during the sampling phase, but is independent of potential random bits used by the algorithm.})
the elements that do not belong to the sample. Whenever an element is observed, it reveals its weight and the algorithm has to decide immediately and irrevocably whether to select it. The set of all selected elements must be in $\mathcal{F}$, and the goal is to maximize the expected weight of the selected elements.
\end{definition}

In the literature, different assumptions have been made about what is known upfront about the ground set $N$ and the constraint $\mathcal{F}$. The weakest assumption that is typically made, is that one only knows the number $n=|N|$ of elements, and that the independence oracle can only be called on elements that appeared so far. However, some algorithms in the literature make the stronger assumption that both $N$ and $\mathcal{F}$ are fully known upfront. We use the weaker assumption as the default assumption for {\MSP}, and we explicitly mention it whenever we rely on an algorithm needing the stronger one. Naturally, our {\MISP} algorithms require access to the same information as required for the {\MSP} algorithms they employ. In other words, when employing for a given matroid an {\MSP} algorithm that relies on the default assumption, our algorithms need only an independence oracle for this matroid which can check the independence of sets consisting of elements that appeared so far; and when employing for a given matroid an {\MSP} algorithm relying on full access to the matroid, then our algorithms also need full access to this matroid.

An important observation that we use is that if an algorithm in the order-oblivious secretary model sets the number $m$ of elements to be observed in the sampling phase randomly according to a binomial distribution $B(|N|,p)$, for some $p\in [0,1]$, then this corresponds to observing each element independently with probability $p$. For brevity, we will sometimes simply say that an algorithm observes each element with probability $p$ during the sampling phase.



%


Finally, whenever we talk about an optimum solution, typically denoted by $\OPT$, we assume that $\OPT$ does not contain any elements of weight $0$. Moreover, for simplicity we often use `$+$' and `$-$' to denote adding/removing a single element to/from a set. For example, $S+u-v = (S\cup\{u\})\setminus \{v\}$.

\section{Making optimal sets of different matroids overlap}\label{sec:makeOptsOverlap}

In this section, we describe a key step in our framework: selecting (online) a weight
function so that the optimal solutions of the different matroids have large
overlap.  As discussed in the introduction, this overcomes the issue that the
optimal solutions of the matroids may have little overlap (if at all). Indeed, once the optimal
solutions overlap, there is a natural  way to combine $\optc$-competitive algorithms for {\MSP} as explained in the next sections.

The setting is as follows. We have $k$ matroids $M_1 = (N, \mathcal{I}_1), M_2
= (N, \mathcal{I}_2), \dots, M_k = (N, \mathcal{I}_k)$ on a common ground set
$N$ whose elements are weighted by $w\colon N \rightarrow \mathbb{R}_{\geq 0}$. Recall that we assume that a consistent tie-breaking rule is used to compare elements of equal weights, and let us number the elements of $N$ by $e_1, \dotsc, e_n$ in such a way that, for every $1 \leq i \leq n-1$, either $w(e_i) > w(e_{i + 1})$ or $w(e_i) = w(e_{i + 1})$ and $e_i$ is considered larger according to the tie breaking rule.
We also let $\mathcal{F} = \mathcal{I}_1 \cap \mathcal{I}_2 \cap \dots \cap
\mathcal{I}_k$ denote the feasible sets in the intersection of the $k$
matroids. 
An important ingredient of our approach is the standard offline
greedy algorithm:
\begin{itemize}
  \item[] Initialize a partial solution $G = \varnothing$.  For $i=1, \dots, n$, if $G + e_i \in \mathcal{F}$ then add $e_i$ to $G$.   
\end{itemize}
The set $G$ returned by the greedy algorithm after iterating through all the elements is
well-known to be a $\sfrac{1}{k}$-approximation when applied to the problem of
finding a maximum weight common independent set in the intersection of $k$
matroids.  When referring to the greedy algorithm we shall use the following notation: 
%
%
%
%
%
\begin{definition}
Let $S\subseteq N$. We denote by $\greedy(S)\subseteq S$ the output of the greedy algorithm when run on the subset $S$ of the elements.  We say that an element $e\in N\setminus S$ is \emph{greedy-relevant} with respect to $S$ if $e\in \greedy(S+e)$. Moreover, we denote by $\imp(S)\subseteq N\setminus S$ all elements that are greedy-relevant with respect to $S$.
\end{definition}

We are now ready to present Algorithm~\ref{alg:reduction} that takes as input a subset $S\subseteq N$ of the elements and  defines a weight function $\wh$ on the remaining elements. We remark that Algorithm~\ref{alg:reduction} can clearly be implemented in an online manner on the elements in $N\setminus S$ since the weight $w'(e)$ of an element $e$ is only  a function of $S$ and $e$.  We show below that, in expectation over the input set $S$, the optimal solutions of $M_1, \dots, M_k$ with respect to $\wh$ have a large overlap.
\begin{algorithm2e}
  \SetKwInOut{Input}{input}\SetKwInOut{Output}{output}
  \DontPrintSemicolon
  \Input{subset $S\subseteq N$}
  \Output{weight function $\wh$}
  \BlankLine
  For each element $e\in S$, define $\wh(e) =0$.\;
  For each other element $e\in N\setminus S$, let 
  $\wh(e) = \begin{cases}
    w(e) & \mbox{if  $e\in \imp(S)$}\,, \\
    0 & \mbox{otherwise}\,.
  \end{cases}
  $
\caption{OverlappingOPT$(S)$}\label{alg:reduction}
\end{algorithm2e}

One should think of the set $S$ supplied to Algorithm~\ref{alg:reduction} as a ``sample''. Indeed, in all our uses of this algorithm the set $S$ is obtained from elements arriving during the sample phase of an order-oblivious algorithm. Note that the ``sampled'' elements of $S$ receive a weight of $0$. This is important as it is
indeed  easy to define a weight function  such that the
optimal solutions  overlap on \emph{already} seen elements. The nice property of
Algorithm~\ref{alg:reduction} is that the weight function $\wh$ defined for each
element $e\in N$ as a function of $S$ is such that the optimal solutions with
respect to $\wh$ have  a large overlap despite containing only unseen elements from $N \setminus S$ that our online algorithm can still select (recall that we assume that the optimal solution does not contain elements of weight $0$).  For further reference, we summarize two useful properties of $\wh$ in the following observation.
\begin{observation}
  \label{obs:weight}
  On input $S\subseteq N$, Algorithm~\ref{alg:reduction} defines a weight
  function $\wh$ satisfying (i) $\wh\leq w$ and (ii) $\wh(e) = 0$ if $e\in S$ or
  $e$ is a loop in one of the matroids (i.e., $\{e\}$ is a dependent set in this matroid).
\end{observation}

We now formally analyze the expected weight of the overlap of the optimal solutions with respect to $\wh$
when $S\subseteq N$ is a randomly chosen subset. Recall that $\wh$ is a function of
the set $S$ (we do not explicitly write this dependency as the set $S$  that
$\wh$ depends on will always be clear from the context).  For $p\in [0,1]$, we
let $\mu_p$ denote the distribution over subsets of $N$ where each element is
included in the subset with probability $p$ independently of other elements.
For notational convenience we also denote, for $j\in [k]$, by $\OPTh_j$ the
(unique) maximum weight independent set in $M_j$ with respect to $\wh$. Notice
that, by definition and since we assume that the optimal solutions do not contain any elements of weight $0$, we have that the optimal independent set of $M_j$ with
respect to $\wh$ equals the optimal solution with respect to $w$ when only
considering elements in $\imp(S)$.  In particular, $\OPTh_j \subseteq
N \setminus S$. We also let $\OPT$ denote the maximum weight solution (with
respect to $w$) in the intersection of the $k$ matroids $M_1, \dots, M_k$. The
main result of this section can now be stated as follows. 
\begin{theorem} 
  For $p=\tfrac{2k-1}{2k}$,
  \begin{align*}
    \E_{S\sim \mu_p} \left[ \wh\left( \cap_{j=1}^k \OPTh_j \right) \right] \geq \frac{1}{4k^2} w(\OPT)\,.
  \end{align*}
(We remark  that $\OPTh_j$ is a  function of $S$ via $\wh$.)
  \label{thm:reduce}
\end{theorem}
As $\wh \leq w$, the above theorem implies that 
    $\E_{S\sim \mu_p} \left[ w\left( \cap_{j=1}^k \OPTh_j \right) \right] \geq \frac{1}{4k^2} w(\OPT)$.
  The theorem follows from the following slightly stronger and more technical lemma. Let $N_{\leq \ell}= \{e_1, \dots, e_\ell\}$ be the set containing the $\ell$ elements of highest weight. 

\begin{lemma}\label{lem:impPrefixGood}
  Let $p\in (0,1]$, and let $G=  \greedy(S)$ and $W = \cap_{j=1}^k \OPTh_j$ be random variables defined by $S\sim \mu_p$. Then,
\begin{equation*}
  \E [|W\cap N_{\leq \ell}|] \geq \frac{(1-(1-p)k)(1-p)}{p} \E[|G\cap N_{\leq \ell}|] \qquad \forall \ell\in [n]\enspace.
\end{equation*}
\end{lemma}
We prove this lemma in the next subsection. Let us now see how it implies Theorem~\ref{thm:reduce}.

\begin{proof}[Lemma~\ref{lem:impPrefixGood} implies Theorem~\ref{thm:reduce}]
  Substituting $p= \tfrac{2k-1}{2k}$ in Lemma~\ref{lem:impPrefixGood} yields,  for every $\ell\in [n]$,
  \begin{align*}
    \E [|W\cap N_{\leq \ell}|] \geq \tfrac{1}{2(2k-1)} \E[|G\cap N_{\leq \ell}|]\,.
  \end{align*}
As the elements of $N$ are ordered in a decreasing weight order according to $w$, the last inequality implies $\E[w(W)] \geq \tfrac{1}{2(2k-1)} \E[w(G)]$. Additionally, we can observe that $\wh(e) = w(e)$ for every element $e \in W$, and thus, $\E[\wh(W)]= \E[w(W)] \geq \tfrac{1}{2(2k-1)} \E[w(G)]$.

We complete the proof by showing that $\E[w(G)] \geq \tfrac{2k-1}{2k^2} w(\OPT) =  \tfrac{p}{k} w(\OPT)$. For this purpose, let $\OPT(S)\subseteq S$ be the maximum weight set in $S$ which is independent in all the matroids. Then, we have
\begin{equation*}
\E[w(\OPT(S))] \geq \E[w(\OPT\cap S)] = p \cdot w(\OPT)\enspace.
\end{equation*}
Moreover, the greedy algorithm is a $k$-approximation algorithm when applied to the intersection of $k$ matroids, which implies
\begin{equation*}
\E[w(G)] \geq \frac{1}{k} \cdot \E[w(\OPT(S))]\enspace,
\end{equation*}
and the desired inequality $\E[w(G)] \geq  \tfrac{p}{k} w(\OPT)$  follows by combining the last two inequalities.
\end{proof}

\subsection*{Proof of Lemma~\ref{lem:impPrefixGood}}
%
%

To prove Lemma~\ref{lem:impPrefixGood}, we consider the offline algorithm given as Algorithm~\ref{alg:simp}, which constructs sets $G$ and $W$ with a joint distribution identical to the random variables $G=\greedy(S)$ and $W=\cap_{j=1}^k \OPTh_j$ in the statement of the lemma (recall that $\OPTh_j$ is the maximum weight independent set of $M_j$ restricted to  greedy-relevant elements).
   \begin{algorithm2e}
        \DontPrintSemicolon
        \SetKwFor{Prob}{With probability}{}{endprob}
        \SetKwFor{Otherwise}{Otherwise}{}{endother}
        $G= \varnothing$, $W = \varnothing$ and $\OPTh_j = \varnothing$ \quad $\forall j\in [k]$\\
								\For{$i=1,\ldots, n$}{
          \If{$G+e_i \in \mathcal{F}$}{
            \Prob{$p$:}{
              $G= G+e_i$\;
            }
            \Otherwise{ (with remaining probability $1-p$):}{
              \For{$j=1, \ldots, k$} { 
                \label{l:rest}    \lIf{$\OPTh_j + e_i \in \mathcal{I}_j$}{$\OPTh_j = \OPTh_j + e_i$}
              } 
              \lIf{$e_i \in \cap_{j=1}^k \OPTh_j $}{$W= W+e_i$}
            }
            
          }
        }
      \caption{$\simp(p)$}\label{alg:simp}
   \end{algorithm2e}
Without loss of generality, we prove Lemma~\ref{lem:impPrefixGood} only for
$\ell=n$. The case of general $\ell$ then follows by considering a restricted
ground set $N$ consisting only of the elements of $N_{\leq \ell}$. Indeed, it is
clear from the description of Algorithm~\ref{alg:simp} that elements in $N
\setminus N_{\leq \ell}$ do not affect the distributions of $G\cap N_{\leq
\ell}$ and $W \cap N_{\leq \ell}$.

For analysis purposes we extend Algorithm~\ref{alg:simp} to maintain additional sets $W'$ and $H_j$ (for $j = 1,...,k$). Algorithm~\ref{alg:simpExt} extends Algorithm~\ref{alg:simp} so as to describe how we maintain these sets (the changes are highlighted in gray for convenience). One can observe that the way each set $H_j$ is maintained, guarantees that throughout the execution of Algorithm~\ref{alg:simpExt} we have $H_j\in \mathcal{I}_j$.

   \begin{algorithm2e}[ht]
        \DontPrintSemicolon
        \SetKwFor{Prob}{With probability}{}{endprob}
        \SetKwFor{Otherwise}{Otherwise}{}{endother}
        $G= \varnothing$, $W = \varnothing$, $W' = \varnothing$, $\OPTh_j = \varnothing$ and $H_j=\varnothing$ \quad $\forall j\in [k]$\;
        \For{$i=1,\ldots, n$}{
          \If{$G+e_i \in \mathcal{F}$}{
            \Prob{$p$:}{
            $G= G+e_i$\;
             \HiLi \For{$j=1,\dots, k$}{
             \HiLi   \lIf{$H_j + e_i \in \mathcal{I}_j$}{$H_j = H_j + e_i$}
             \HiLi  \If{$H_j + e_i \not \in \mathcal{I}_j$} {\HiLi  $H_j = H_j - f + e_i$, where $f\in H_j\setminus G$ is any element with $H_j -f +e_i\in \mathcal{I}_j$. \label{step:exchange}}
              }
            }
            \Otherwise{ (with remaining probability $1-p$):}{
              \For{$j=1, \ldots, k$} { 
                \lIf{$\OPTh_j + e_i \in \mathcal{I}_j$}{$\OPTh_j = \OPTh_j + e_i$}  \label{step:easy} 
                \HiLi  \lIf{$H_j + e_i \in \mathcal{I}_j$}{$H_j = H_j + e_i$} \label{step:rest}
              } 
              \lIf{$e_i \in \cap_{j=1}^k \OPTh_j $}{$W= W+e_i$}
              \HiLi  \lIf{$e_i \in \cap_{j=1}^k H_j $}{$W'= W'+e_i$}
            }

          }
        }

    \caption{$\simpExt(p)$}\label{alg:simpExt}
   \end{algorithm2e}

%
In the following, we denote by $G^i$ and $H_j^i$ the sets $G$  and $H_j$ at the end of iteration $i\in[n]$ of the algorithm.
Maintaining that $H_j$ for $j\in [k]$ is independent in $M_j$ throughout the algorithm can thus be rephrased as $H_j^i\in \mathcal{I}_j$ for $i\in [n]$.
Additionally, one can observe that the way $H_j$ is
maintained also implies that the span of $H_j$ equals the span of $G \cup \OPTh_j$. Thus, we have $\OPTh_j + e_i \in \mathcal{I}_j$ 
at Step~\ref{step:easy} of Algorithm~\ref{alg:simpExt} whenever we have
$H_j + e_i \in \mathcal{I}_j$ at Step~\ref{step:rest}, and therefore, the constructed set $W'$ is always a subset of $W = \cap_{j=1}^k
\OPTh_j$. We proceed to prove the inequality $\E[| W'|]
\geq \frac{(1-(1-p)k)(1-p)}{p} \E[|G|]$, which is stronger than the statement of the lemma.

For $i\in [n]$ we define the events
\begin{align*}
A_i  &\coloneqq  (G^{i-1} + e_i\in \mathcal{F}) \wedge (H_j^{i-1}+e_i \in \mathcal{I}_j \;\;\forall j\in [k]) \enspace, \\
B_i &\coloneqq  (G^{i-1} + e_i \in \mathcal{F}) \wedge \neg (H_j^{i-1}+e_i \in \mathcal{I}_j \;\;\forall j\in [k]) \enspace,
\end{align*}
and let
\[
\alpha_i \coloneqq \Pr[A_i]
\qquad \text{and} \qquad
\beta_i \coloneqq \Pr[B_i]\enspace.
\]

The following are two basic observations related to the above quantities.
\begin{equation} \label{eq:expG_expW}
\E\left[|G|\right] =    p \sum_{i=1}^n (\alpha_i + \beta_i)
\qquad \text{and} \qquad
\E\left[|W'|\right] = (1-p)\sum_{i=1}^n \alpha_i\enspace.
\end{equation}

We further relate the above quantities by providing an upper bound on $\E[\sum_{j=1}^k |H_j \setminus G|]$. More precisely, we will consider how $\sum_{j=1}^k |H_j \setminus G|$ changes over the iterations of the algorithm. The intuition why this quantity is of interest is as follows. Suppose that at some point during the algorithm $H_j = G$ for all $j\in [k]$, then the next element $e$ such that $G+e \in \mathcal{F}$ is put in $G$ with probability $p$ and put in $W'$ with probability $1-p$. This indicates that it is desirable for the quantity $\sum_{j=1}^k |H_j \setminus G|$ to be small. That we can  upper bound this quantity in expectation follows from the fact that, if there is an element for which $G+e \in \mathcal{F}$ but $e$ would not be added to $W'$ because $H_j + e \not \in \mathcal{I}_j$ for some $j \in [k]$, then we have (due to Step~\ref{step:exchange}) that the quantity $\sum_{j=1}^k |H_j \setminus G|$ decreases in expectation if $p$ is chosen appropriately. We continue with the formal analysis. For $i = 0,\ldots,n$ we define 
\begin{align*}
X_i &\coloneqq \sum_{j=1}^k |H_j^i \setminus G^i|\enspace,
\end{align*}
where $X_0 = 0$.
To obtain a bound on $\E[\sum_{j=1}^k |H_j \setminus G|] = \E[X_n]$, we study the changes $\E[X_i] - \E[X_{i-1}]$.
First, notice that if neither event $A_i$ happens nor $B_i$, then $G^{i-1}+e_i \not\in \mathcal{F}$, and the change $X_i - X_{i-1}$ is zero. Hence, by the law of total expectation, we have for $i = 1,\dotsc, n$
\begin{align}
\E[X_i - X_{i-1}] &= \alpha_i \cdot \E[X_i - X_{i-1} \mid A_i]
+ \beta_i \cdot \E[X_i - X_{i-1} \mid B_i]\enspace.
\label{eq:diffXTotExp}
\end{align}
Furthermore, for every such $i$ it holds that
\begin{align}
\E[X_i - X_{i-1} \mid A_i] &=  (1-p) k\enspace, \text{ and}\label{eq:diffXCondA}\\
\E[X_i - X_{i-1} \mid B_i] &\leq  (1-p) k - 1 \enspace\label{eq:diffXCondB},
\end{align}
due to the following. Let
$\Gamma = \{j\in [k] \mid H_j^{i-1}+e_j \not\in \mathcal{I}_j\}$.
Hence, $A_i$ corresponds to $\Gamma = \varnothing$, whereas $B_i$ corresponds to $\Gamma \neq \varnothing$. If $e_i$ gets added to $G$, which happens with probability $p$, then in each $H_j^{i-1}$ for $j\in \Gamma$ one element of $H_j^{i-1}\setminus G^{i-1}$ is replaced by $e_i$ to obtain $H_j^{i}$. Moreover, if $e_i$ does not get added to $G$, which happens with probability $1-p$, then $e_i$ is added to each $H_j^{i-1}$ with $j\in [k]\setminus \Gamma$ to obtain $H_j^i$. Hence, overall we have
\begin{equation*}
\E[X_i - X_{i-1}] = p\cdot (-|\Gamma|) + (1-p)\cdot (k-|\Gamma|) = (1-p)k - |\Gamma| \enspace,
\end{equation*}
where the expectation is only over the decision whether to add $e_i$ to $G$ or not. Equation~\eqref{eq:diffXCondA} now follows by plugging $\Gamma=\varnothing$ into the above equality, and~\eqref{eq:diffXCondB} follows by observing that, for $\Gamma\neq\varnothing$, the above expression is maximized for $|\Gamma|=1$.

We thus obtain
\begin{align*}
0 &\leq \E[X_n] = \sum_{i=1}^n \E[X_i - X_{i-1}] \\
  &\leq (1-p)k \left(\sum_{i=1}^n \alpha_i\right)
     + ((1-p)k-1) \left(\sum_{i=1}^n \beta_i\right) && \text{(using~\eqref{eq:diffXTotExp}, \eqref{eq:diffXCondA}, and \eqref{eq:diffXCondB})}\\
  &=  (1-p)k \left(\sum_{i=1}^n \alpha_i + \beta_i\right)  - \sum_{i=1}^n \beta_i\\
  &\leq \frac{1-p}{p} k \E[|G|] - \left(\frac{1}{p}\E[|G|] - \frac{1}{1-p} \E[|W'|] \right)
 &&\text{(using~\eqref{eq:expG_expW})}\\
  &= \frac{(1-p)k-1}{p} \E[|G|] + \frac{1}{1-p}\E[|W'|]\enspace.
\end{align*}

Finally, by reordering terms (and using our previous observation that $W' \subseteq W$, and thus, $\E[|W|] \geq \E[|W'|]$), we obtain the statement of Lemma~\ref{lem:impPrefixGood} for $\ell=n$ as desired. Recall that this implies the lemma for all $\ell\in [n]$ as discussed in the beginning of the proof.

\section{Extending \texorpdfstring{$\OPT$}{OPT}-selecting \MSP algorithms to the intersection of matroids}\label{sec:optSelectToMatInt}
In this section, we  use the reduction presented in the previous section to
combine algorithms for {\MSP}. In order to achieve
a good guarantee for the intersection of the matroids we make two assumptions
on the considered single-matroid algorithms. As already discussed, we 
consider algorithms that are order-oblivious, and, in addition, we shall introduce
the assumption of \emph{$\optc$-competitiveness}.  These assumptions are easily
satisfied by existing algorithms for {\MSP} when the
matroid is a  partition, laminar or co-graphic matroid. To capture other
algorithms in the literature (that are not $\optc$-competitive) such as algorithms for graphic matroids, transversal
matroids and the best known algorithm for general matroids, we generalize the
framework in the next section. As the proofs in this section are rather clean,
we believe that it serves as a good starting point and motivation before
reading the more complex (and general) framework.  In what follows, we first
define $\optc$-competitiveness, and then state and prove the main theorem of this
section.

\begin{definition}[$\optc$-competitiveness]
  We say that an order-oblivious algorithm $\mathcal{A}$ is $c$-$\optc$-competitive if for any $e\in \OPT$,
  \begin{align*}
    \mathbb{E}_{S} \left[ \min_{\sigma}  \Pr_{r} [ e\in \mathcal{A}(S, r, \sigma)] \right] \geq c\,,
  \end{align*}
  where $\mathcal{A}(S,r, \sigma)$ denotes the elements selected by
  $\mathcal{A}$ when using the random bits $r$, given the sample $S$, and where
  the elements of $N \setminus S$ arrive in the second phase according to  the
  adversarial order $\sigma$ that is allowed to depend on $S$ and the considered element $e$. 
\end{definition}
We remark that a $c$-$\optc$-competitive algorithm is clearly $c$-competitive, but the other direction may not hold.
At first sight, $\optc$-competitiveness perhaps seems like a less severe
restriction than order-obliviousness. However, it turns out that most\footnote{In fact all current algorithms which are known to the authors can be made order-oblivious by losing at most a small constant factor in the competitive ratio.}  \MSP algorithms are
order-oblivious, but many are only $c$-competitive and not $c$-$\optc$-competitive. The
generalization of the framework in the next section is therefore designed so as to relax
the assumption of $\optc$-competitiveness while we keep the rather benign
assumption of order-obliviousness. 

Having defined $\optc$-competitiveness, the result of this section follows rather easily from our reduction in the previous section. 
\begin{reptheorem}{thm:mainSimple}
  Consider $k$ matroids $M_1,\dots, M_k$ on a common ground set $N$. Suppose that, for $j=1,\dots, k$,
  there is an order-oblivious $c_j$-$\optc$-competitive \MSP algorithm
  $\mathcal{A}_j$ on matroid $M_j$. Then, there is an order-oblivious $(\frac{1}{4k^2}\prod_{j=1}^k c_j)$-competitive online
  algorithm for {\MISP} on the intersection $M_1 \cap M_2 \cap
  \dots \cap M_k$.
\end{reptheorem}
\begin{proof}
  Let $\mathcal{I}_j$ denote the family of independent sets of matroid $M_j$. 
  We describe an order-oblivious algorithm for {\MISP} on the intersection $M_1 \cap M_2 \cap \dots \cap M_k$. At first we describe the sample phase of the algorithm (Steps~\ref{step:first_sample} and~\ref{step:second_sample} below) as if we could take several independent samples from the set $N$, and we then explain how we can implement these steps by taking a single sample. The algorithm proceeds as follows:
  \begin{enumerate}
    \item Let $p = \frac{2k-1}{2k}$, and let $S$ be a sample containing every element of $N$ with probability $p$, independently. Observe that $S \sim \mu_p$, and thus, we can use Algorithm~\ref{alg:reduction} on $S$ to obtain a weight function $w'$ on all the elements of $N$ that satisfies the properties of Observation~\ref{obs:weight} and Theorem~\ref{thm:reduce}.
\label{step:first_sample}
    \item For each $j=1,\dots, k$, 
          let $m_j$ be the number of elements $\mathcal{A}_j$ samples, and let $S_j \subseteq N$ be a uniformly at random sample from $N$ of
          cardinality $m_j$ (if the number of elements $\mathcal{A}_j$ samples is random, let $m_j$ be the outcome in the current execution of $\mathcal{A}_j$). \label{step:second_sample}
    \item Finally, let $I=\varnothing$, and  start the selection phase of the algorithms $\mathcal{A}_1, \dots, \mathcal{A}_k$ with the weight function $w'$. For each arriving remaining element (i.e., those elements that were not in any sample), add it to our solution (i.e., set $I\leftarrow I + e$) if each one of the algorithms $\mathcal{A}_1, \dots, \mathcal{A}_k$ selects element $e$.\label{step:selection}
  \end{enumerate}

  We remark that the output $I$ is an independent set in the intersection of $M_1, \dots, M_k$ since each $\mathcal{A}_j$ selects a set $I_j$ that is independent in $M_j$ and $I = \cap_{j=1}^k I_j$.  
  
  Before analyzing the expected weight of $I$, we explain how to implement the
  algorithm in an order-oblivious fashion.
  \begin{claim*}
    The above algorithm can be implemented in the order-oblivious model.
  \end{claim*}
  \begin{proof}
    To implement the algorithm in the order-oblivious model, we need to turn the several samples made in Steps~\ref{step:first_sample} and~\ref{step:second_sample} into a single sample during the sample phase (and leave the selection phase, i.e.,  Step~\ref{step:selection}, untouched).  Let $m = |S|$ and $m_j = |S_j|$ for $j=1,
  \dots, k$ be the cardinality  of the sets. In addition, let $q_1 = |S_1 \cap
  S|, q_2 = |S_2 \cap (S \cup S_1)|, \dots, q_k = |S_k \cap (S\cup S_1 \cup
  \dots \cup S_{k-1})|$, i.e., $q_j$ denotes cardinality of the intersection of
  $S_j$ with  $S \cup S_1 \cup \dots \cup S_{j-1}$. Notice that the execution of
  Steps~\ref{step:first_sample} and~\ref{step:second_sample} produces a distribution $D$ over integers $m, m_1, \dots, m_k, q_1,
  \dots, q_k$ and an equivalent way of sampling $S, S_1, \dots, S_k$ would be to first sample $m, m_1, \dots, m_k, q_1, \dots, q_k$ according to $D$, and then obtain $S$ by selecting $m$ elements  uniformly at random from $N$, $S_1$ by selecting $q_1$ elements uniformly at random from $S$ and $m_1 -q_1$ elements from $N\setminus S$, $S_2$ by selecting $q_2$ elements uniformly at random from $S \cup S_1$ and $m_2 - q_2$ elements uniformly at random from $N\setminus (S \cup S_1)$, and so on. We can thus implement the sampling phase order-obliviously as follows.
  \begin{itemize}
    \item Sample $m, m_1, \dots, m_k, q_1, \dots, q_k$ according to distribution $D$.
    \item Take $A \subseteq N$ to be a uniformly at random set of cardinality $m+ q_1 + q_2 + \dots q_k$ (this is our single sample of the elements in the sample phase).
    \item Obtain $S$ by selecting $m$ elements from $A$ uniformly at random. Additionally, for every $j=1, \dots, k$, obtain $S_j$ by selecting $q_j$ elements from $S \cup S_1 \cup \dots \cup S_{j-1}$  and $m_j -q_j$ elements from $A \setminus (S \cup S_1 \cup \dots \cup S_{j-1})$ uniformly at random. \qedhere
  \end{itemize}
\end{proof}


Having explained how to implement the algorithm,   we proceed to analyze the expected weight $w(I) \geq w'(I)$ of the returned set. By Theorem~\ref{thm:reduce}  in the previous section, we have  (using that $w \geq w'$)
  \begin{align}
    \label{GuaranteeAfterFirstStep}
    \mathbb{E}[w(\cap_{j=1}^k \OPT'_j)] \geq \frac{1}{4k^2}w(\OPT)\enspace,
  \end{align}
  where the expectation is taken over the sample $S \sim \mu_p$. Recall that $\OPT'_j$ denotes the
  maximum weight independent set in $M_j$ with respect to weight function $w'$ (not including elements of weight $0$),
  and $\OPT$ denotes the maximum weight independent set in the intersection
  with respect to weight function $w$. Let $U = N\setminus S$. We remark that, as $w'(e) = 0$ for
  every element $e\in S$ (see Observation~\ref{obs:weight}), $\OPT'_j \subseteq U$. 
  
  We shall analyze the expected weight $w(I)$ by showing that, for a fixed sample  $S$, each element in
  $\cap_{j=1}^k \OPT'_j$ is taken with a good probability:
  \begin{claim*}
    For each element $e \in \cap_{j=1}^k \OPT'_j$, 
    \begin{align*}
      \Pr[e\in I] \geq  \prod_{j=1}^k c_j\,, 
    \end{align*}
    where the probability is over the sets $S_1, \dots, S_k$ and the randomness used by the algorithms $\mathcal{A}_1, \dots, \mathcal{A}_k$.
  \end{claim*}
  \begin{proof}
    By definition,
    \begin{align*}
      \Pr[e\in I] =  \Pr_{S_1,r_1, \dots, S_k, r_k, \sigma} [ e\in \cap_{j=1}^k \mathcal{A}_j(S_j, r_j, \sigma)]\,, 
    \end{align*}
    where $\sigma$ is the given order of the elements seen in the selection phase (Step~\ref{step:selection}), i.e., of the elements $U \setminus (S_1 \cup \dots \cup S_k)$.\footnote{Strictly speaking,  $\mathcal{A}_j$ expects to see $n- m_j$ elements in the selection phase. This can be achieved by simply feeding the elements of $(S  \cup S_1 \cup\dots \cup  S_k) \setminus S_j$ to $\mathcal{A}_j$ in any order at the end (without including any of them in our solution $I$).}
    We can clearly lower bound this probability by selecting the worst-case order (even allowing a different order for each algorithm):
    \begin{align*}
      \Pr[e\in I] \geq  \mathbb{E}_{S_1, \dots, S_k} \left[ \min_{\sigma_1, \dots, \sigma_k} \Pr_{r_1, \dots, r_k} [ e\in \cap_{j=1}^k \mathcal{A}_j(S_j, r_j, \sigma_j)]\right]\,, 
    \end{align*}
    where $\sigma_j$ is an ordering of the elements in $U\setminus S_j$ that is allowed to depend on $S_1,  \dots, S_k$ and $e$. The beauty of this lower bound is that now the execution of each algorithm becomes independent. Indeed,
    \begin{align*}
      \mathbb{E}_{S_1, \dots, S_k} \left[ \min_{\sigma_1, \dots, \sigma_k} \Pr_{r_1, \dots, r_k} [ e\in \cap_{j=1}^k \mathcal{A}_j(S_j, r_j, \sigma_j)]\right]
      & = 
      \mathbb{E}_{S_1, \dots, S_k} \left[ \min_{\sigma_1, \dots, \sigma_k} \prod_{j=1}^k \Pr_{r_j} [ e\in  \mathcal{A}_j(S_j, r_j, \sigma_j)]\right] \\
      & = 
      \mathbb{E}_{S_1, \dots, S_k} \left[  \prod_{j=1}^k \min_{\sigma_j}\Pr_{r_j} [ e\in  \mathcal{A}_j(S_j, r_j, \sigma_j)]\right] \\
      & = 
         \prod_{j=1}^k \mathbb{E}_{S_j}\left[\min_{\sigma_j}\Pr_{r_j} [ e\in  \mathcal{A}_j(S_j, r_j, \sigma_j)]\right]
      \geq 
         \prod_{j=1}^k c_j \enspace,
    \end{align*}
    where the last equality follows since the samples $S_1, \dots, S_k$ are independent and the  
    inequality follows since $\mathcal{A}_j$ is $c_j$-$\optc$-competitive (and $e\in \OPT'_j \subseteq U$). 
  \end{proof}
  By linearity of expectation, the above claim combined with~\eqref{GuaranteeAfterFirstStep} implies the theorem. 
\end{proof}

\noindent \textbf{Remark:} The competitive ratio guaranteed by Theorem~\ref{thm:mainSimple} (and our more general framework presented in Section~\ref{sec:genToMatInt}) is exponential in $k$ even when $c_1, \dotsc, c_k$ are all constants. This is the result of the fact that we need to execute independently the algorithms for the individual matroids of the intersection. In many cases this can be avoided using the FKG inequality. More specifically, many algorithms for individual matroids have the property that the probability of an element $e \in OPT \setminus S$ to be selected can only increase as other elements are added to $S$, and moreover, this probability approaches $1$ as the sample becomes very large. For such algorithms one can use a common very large sample, and get an algorithm which selects every given element in the intersection of the optimal solutions of the individual matroids with a significant probability. We omit the technical details, but note that this mode of employment of our ideas resembles the contention resolution schemes framework of~\cite{chekuri_2014_submodular}. Moreover, we would like to stress that there are examples of algorithms for {\MSP} which do not have the above property---the most notable of which are the algorithms for general matroids of~\cite{babaioff_2007_matroids,chakraborty_2012_improved,feldman_2015_simple,lachish_2014_competitive}.

\section{A generalized framework to extend \MSP algorithms to \MISP}\label{sec:genToMatInt}

In Section~\ref{sec:optSelectToMatInt} we have seen a simple version of our framework that works for algorithms that are $\optc$-competitive. Unfortunately, many known algorithms for {\MSP} are not $\optc$-competitive, and there is no obvious way to make them $\optc$-competitive. Definition~\ref{def:weakly_opt_competitive} relaxes $\optc$-competitiveness in a way that makes it include algorithms that are not $\optc$-competitive due to two common issues. The first issue is that some algorithms occasionally select high weight non-$\OPT$ elements early, and these elements then block them from selecting some $\OPT$ elements later. Intuitively, this should not be an issue since the algorithm is guaranteed to select good elements instead of the elements of $\OPT$, however, this violates $\optc$-competitiveness. The second issue is that some algorithms ignore very light elements of the optimal solution whose total contribution to the solution is minor.

Definition~\ref{def:weakly_opt_competitive} is based on a slight extension of the order-oblivious secretary model which is presented by the following definition.

\begin{definition}
We occasionally assume the existence of a ``switching adversary'' with the following power. Consider an execution of an order-oblivious algorithm $\cA$ for {\MSP} on matroid $M = (N, \cI)$. Every time after $\cA$ selects an element $e \in N$, the switching adversary immediately and irrevocably puts it in exactly one of two ``sub-solutions'' named $T^\cA_1$ and $T^\cA_2$. The algorithm learns the decision of the switching adversary immediately after it is made, and must keep the sub-solution $T^\cA_1$ independent in $\cI$, i.e., it cannot select an element $e$ if its addition to $T^\cA_1$ would violate $T^\cA_1$'s independence because the switching adversary may decide to add $e$ to this set. In contrast, the set $T^\cA_2$ (and also the union $T^\cA_1 \cup T^\cA_2$) need not be kept independent.
\end{definition}

Intuitively, one can think of the switching adversary as an adversary with the power to partially discard elements from the solution of the algorithm by assigning them to the sub-solution $T^\cA_2$. In the context of $\OPT$-competitiveness, we want this discard to be meaningful only for non-$\OPT$ elements. In other words, $\OPT$ elements that are discarded by the switching adversary should still count towards the algorithm's objective. The following definition formalizes this intuitive notion.

 \begin{definition} \label{def:weakly_opt_competitive}
Consider an order-oblivious algorithm $\cA$ for {\MSP} on matroid $M = (N, \cI)$ with weight function $w$ which is executed against a switching adversary.
Let $h(\OPT)$ be the heaviest element in $\OPT$, let $L(\OPT, \ell) = \{e \in \OPT \mid w(e) \leq \ell\}$ and let $L(\OPT) = L(\OPT, \ell^*)$---where $\ell^*$ is the maximum value in $\{w(e) \mid e \in N\}$ for which $w(L(\OPT, \ell^*)) < w(h(\OPT))$. Note that $L(\OPT)$ is not well-defined when $|\OPT| \leq 1$, and thus, we assume $L(\OPT) = \varnothing$ in this case . Using these definitions, we say that $\cA$ is \emph{$(c^o, c^a)$-weakly-$\optc$-competitive} if, for every given switching adversary, there exists a random set $O^\cA(S, r, \sigma) \subseteq \OPT$, which is a function of the sample $S$, the random bits $r$ of the algorithm, and the arrival order $\sigma$ of the elements in the non-sample phase such that
\[
	\E_S\left[\min_{\sigma} \Pr_r[e \in O^\cA(S, r, \sigma)]\right] \geq c^o \quad \forall\; e \in \OPT \setminus L(\OPT)
\]
and
\[
	w((T^\cA_1 \cup T^\cA_2) \cap O^\cA(S, r, \sigma)) + c^a \cdot w(T^\cA_1 \setminus O^\cA(S, r, \sigma)) \geq w(O^\cA(S, r, \sigma))
	\enspace.
\]
\end{definition}

Given the above discussed intuition, it is relatively easy to see that any $c$-$\optc$-competitive algorithm $\cA$ is also $(c, 0)$-weakly-$\optc$-competitive with the set $O^\cA(S, r, \sigma)$ chosen as the set of elements selected by $\cA$ given the sample $S$, the random bits $r$ and the arrival order $\sigma$. Conversely, any $(c^o, c^a)$-weakly-$\optc$-competitive algorithm $\cA$ is $\frac{c^o}{2\max\{c^a, 1\}}$-competitive in the absence of a switching adversary because one can think of the absence of a switching adversary as equivalent to a switching adversary placing all elements of $\cA$'s solution into $T^\cA_1$. 

Another class of algorithms that are not $\optc$-competitive consists of algorithms that work by reducing their input problem into {\MSP} on a simpler matroid for which there exists an $\optc$-competitive algorithm. The next definition captures a wide range of algorithms of this kind. Our framework applies to any algorithm that can be made to fit into this definition.
\begin{definition}
A \emph{$(c^r, c^o, c^a)$-reduce-and-solve} algorithm $\cA$ for {\MSP} on matroid $M = (N, \cI)$ is a collection of the following three entities.
\begin{itemize}
	\item A (possibly random) ground set $N^\cA$ with a function $g^\cA\colon N^\cA \to N$ mapping every element of $N^\cA$ to a \emph{source} element in $N$. We call $N^\cA$ a refinement of $N$, and say that the elements of $(g^{\cA})^{-1}(e) \subseteq N^\cA$ refine the element $e$ of $N$. 
	Moreover, it must be possible to evaluate $(g^{\cA})^{-1}(e)$ for every element $e$ that has been observed, and, for technical reasons, every element of $N$ must have at least one refinement element in $N^\cA$.
	\item A (possibly random) matroid $M^\cA = (N^\cA, \cI^\cA)$ such that for every independent set $S$ of $M$ there exists a random set $R^\cA(S)$ which is independent in $M^\cA$ and obeys
	\[
		\Pr[R^\cA(S) \cap (g^{\cA})^{-1}(e) \neq \varnothing] \geq c^r \quad \forall\; e \in S 
		\enspace.
	\]
	Moreover, it must be possible to evaluate independence oracle queries with respect to $M^\cA$ for subsets containing only refinements of elements of $N$ that have already been observed.
	\item A $(c^o, c^a)$-weakly-$\optc$-competitive algorithm $\bar{\cA}$ for {\MSP} on restrictions of $M^\cA$ to subsets of $N^\cA$ containing exactly one refinement element of every source element of $N$. Moreover, this algorithm must have the extra property that $g^\cA(T^{\bar{\cA}}_1)$ is independent in the original matroid $M$.\footnote{In fact, we can even allow $T^{\bar{\cA}}_1$ to be dependent in $M^{\cA}$ as long as $g^\cA(T^{\bar{\cA}}_1)$ is independent in $M$. However, this is not necessary for any of the algorithms that are currently known to fall into our framework.}
\end{itemize}
\end{definition}

Technically, a $(c^r, c^o, c^a)$-reduce-and-solve algorithm on matroid $M$ is not really a secretary algorithm on $M$. However, our framework shows that it implies an $\Omega(c^r c^o / \max\{c^a, 1\})$-competitive secretary algorithm on $M$. Moreover, every \emph{$(c^o, c^a)$-weakly-$\optc$-competitive} algorithm $\cA$ can be cast as a $(1, c^o, c^a)$-reduce-and-solve algorithm by simply setting $N^\cA = N$, $M^\cA = M$ and $\bar{\cA} = \cA$ (with the function $g^\cA$ mapping elements of $N^\cA$ to $N$ being the identity function). We now state our main theorem.

\begin{theorem}
  Consider $k$ matroids $M_1,\dots, M_k$ with a common ground set $N$. Suppose that, for $i=1,\dots, k$,
  there is a $(c^r_i, c^o_i, c^a_i)$-reduce-and-solve algorithm
  $\mathcal{A}_i$ on matroid $M_i$. Then, there is an order-oblivious
  algorithm for {\MISP} on the intersection $M_1 \cap M_2 \cap
  \dots \cap M_k$ which is $\left(\frac{\prod_{i=1}^k c^r_i c^o_i}{8k(k + 1) \cdot \max\{\sum_{i = 1}^k c^a_i, 1\}}\right)$-competitive.
  \label{thm:main}
\end{theorem}

In the rest of this section our objective is to prove Theorem~\ref{thm:main}. We begin with the following simple observation. Let $R = \bigcap_{i = 1}^k g^{\cA_i}(R^{\cA_i}(\OPT))$.

\begin{observation} \label{obs:R_value}
The expected weight of the set $R$ is at least $w(\OPT) \cdot \prod_{i = 1}^k c^r_i$.
\end{observation}
\begin{proof}
Fix an arbitrary element $e \in \OPT$. By definition, $R^{\cA_i}(\OPT)$ contains a refinement of $e$ with probability at least $c^r_i$, which implies that $g^{\cA_i}(R^{\cA_i}(\OPT))$ contains $e$ with at least this probability. Since the membership of $e$ in each set $g^{\cA_i}(R^{\cA_i}(\OPT))$ is independent, we get that the probability that $e$ belongs to $\bigcap_{i = 1}^k g^{\cA_i}(R^{\cA_i}(\OPT))$ is at least $\prod_{i = 1}^k c^r_i$. The observation now follows by the linearity of expectation.
\end{proof}

The following lemma generalizes results from Section~\ref{sec:makeOptsOverlap}. The proof of this lemma is based on a quite technical adaptation of the proofs of Observation~\ref{obs:weight} and Theorem~\ref{thm:reduce}, and is, thus, deferred to the next subsection.

\begin{lemma} \label{lem:intersection_OPT_general}
There is an algorithm that given the matroids $M^{\cA_1}, M^{\cA_2}, \dotsc, M^{\cA_k}$ and a set $S \sim \mu_p$ for $p = \frac{2k - 1}{2k}$ produces a weight function $\wh$ for all the elements of $N$ and a function $d_i \colon N \to N^{\cA_i}$ for every $i = 1,\dotsc, k$ such that
\begin{itemize}
	\item $d_i$ assigns each element of $N$ to one of its refinements in $N^{\cA_i}$.
	\item $w'(e) \leq w(e)$ for every $e \in N$, and $w'(e) = 0$ whenever $e \in S$ or $d_i(e)$ is a loop in $M^{\cA_i}$ for some $i = 1, \dotsc, k$.
	\item For every $i = 1,\dotsc, k$, let $\OPT'_i$ be the optimal set in the matroid $M^{\cA_i}|_{d_i(N)}$ with respect to a weight function assigning to each element $e$ of $M^{\cA_i}|_{d_i(N)}$ the weight $w'(g^{\cA_i}(e))$. Then,\[\E\left[w'\mathopen{}\left(\bigcap_{i = 1}^k g^{\cA_i}(\OPT'_i)\right)\right] \geq \frac{w(R)}{4k(k + 1)} \enspace,\]
	where the expectation is over the randomness of $S$.
\end{itemize}
\end{lemma}

We can now present the algorithm that we use to prove Theorem~\ref{thm:main}. The following description of the algorithm assumes that the algorithm can obtain multiple independent samples of elements from $N$. This assumption can be lifted using the same technique used to remove this assumption from the algorithm presented in Section~\ref{sec:optSelectToMatInt}.
\begin{enumerate}
    \item Let $p = \frac{2k-1}{2k}$, and let $S$ be a sample containing every element of $N$ with probability $p$, independently. Observe that $S \sim \mu_p$, and thus, we can use the algorithm whose existence is guaranteed by Lemma~\ref{lem:intersection_OPT_general} on $S$ to obtain the weight function $w'$ and the functions $d_i$ whose properties are given by the lemma.
		
    \item For each $i=1,\dots, k$, let $m_i$ be the number of elements $\bar{\cA_i}$ samples, and let $S_i \subseteq N$ be a uniformly random sample from $N$ of cardinality $m_i$ (if the number of elements $\bar{\cA_i}$ samples is random, let $m_i$ be the outcome in the current execution of $\bar{\cA_i}$). Then, for every such $i$, pass $d_i(S_i)$ as the sample to algorithm $\bar{\cA_i}$ .
		
    \item Finally, let $I=\varnothing$ be the output set of this algorithm, and start the selection phase of algorithms $\bar{\cA}_1, \dots, \bar{\cA}_k$. For each arriving element $e$ which has not been seen yet (i.e., $e$ does not belong to any sample), pass $d_i(e)$ to the algorithm $\bar{\cA}_i$ for every $i = 1,\dotsc, k$ with the weight $w'(e)$. If all the algorithms $\bar{\cA}_1, \dots, \bar{\cA}_k$ select the elements they got, then add $e$ to $I$ and make the switching adversary of each algorithm $\bar{\cA}_i$ put $d_i(e)$ into the set $T^{\bar{\cA}_i}_1$. In contrast, if at least one of the algorithms $\bar{\cA}_1, \dots, \bar{\cA}_k$ rejects the element it got, then do not add $e$ to $I$ and make the switching adversary of every algorithm $\bar{\cA}_i$ that accepts $d_i(e)$ put $d_i(e)$ into the set $T^{\bar{\cA}_i}_2$. \label{step:feeding}
\end{enumerate}

One can observe that for every $i = 1,\dotsc, k$ the algorithm $\bar{\cA}_i$ is executed on a restriction of $M^{\cA_i}$ to the subset $d_i(N)$. This means that, by the definition of a reduce-and-solve algorithm, that the sets $g^{\cA_1}(T^{\bar{\cA_1}}_1), \dotsc, g^{\cA_k}(T^{\bar{\cA_k}}_1)$ are independent in the matroids $M_1, \dotsc, M_k$, respectively. Moreover, one can observe that all these sets are in fact equal to $I$, and thus, $I$ is independent in the intersection of these matroids. Thus, to prove Theorem~\ref{thm:main} it only remains to analyze the competitive ratio of $I$.

Let $O = \bigcap_{i = 1}^k g^{\cA_i}(O^{\bar{\cA}_i}(d_i(S_i), r_i, \sigma_i))$, where $r_i$ denotes the random bits of $\bar{A}_i$ and $\sigma_i$ is the order in which the elements are fed to algorithm $\bar{\cA_i}$ in Step~\ref{step:feeding} of the above algorithm. It is useful to observe that every set $B \subseteq d_i(N)$ obeys $d_i(g^{\cA_i}(B)) = B$. In other words, $d_i$ is the inverse function of $g^{\cA_i}$ when the domain of $g^{\cA_i}$ is restricted to elements of $d_i(N)$. Additionally, we define the weight according to $w'$ of an element in $N^{\cA_i}$ as equal to the weight of its source element in $N$ according to this weight function. Notice that this definition is consistent with the weights passed to the algorithms $\bar{\cA}_1, \dotsc, \bar{\cA}_k$.

\begin{lemma} \label{lem:O_value}
$\E[w'(O)] \geq \frac{1}{2}\prod_{i = 1}^k c^o_i \cdot w'\mathopen{}\left(\bigcap_{i = 1}^k g^{\cA_i}(\OPT'_i)\right)$, where the expectation is over the samples $S_1, \dotsc, S_k$ and the random bits $r_1, \dotsc, r_k$.
\end{lemma}
\begin{proof}
Consider an arbitrary element $e \in \bigcap_{i = 1}^k g^{\cA_i}(\OPT'_i \setminus L(\OPT'_i))$, and observe that
\begin{align*}
	\Pr[e \in O]
	={} &
	\Pr\left[e \in \bigcap_{i = 1}^k g^{\cA_i}(O^{\bar{\cA}_i}(d_i(S_i), r_i, \sigma_i))\right]\\
	={} &
	\Pr\left[\bigwedge_{i = 1}^k e \in g^{\cA_i}(O^{\bar{\cA}_i}(d_i(S_i), r_i, \sigma_i))\right]
	=
	\Pr\left[\bigwedge_{i = 1}^k d_i(e) \in O^{\bar{\cA}_i}(d_i(S_i), r_i, \sigma_i)\right]
	\enspace.
\end{align*}
The rightmost hand side of the above equality can be lower bounded by replacing the orders $\sigma_1, \dotsc, \sigma_k$ induced by Step~\ref{step:feeding} of our algorithm with worst case orders that are allowed to depend on $S_1, \dotsc, S_k$ and $e$ (but not on the random bits $r_1, \dotsc, r_k$). Thus, we get
\[
	\Pr[e \in O]
	\geq
	\E_{S_1, \dotsc, S_k} \left[\min_{\sigma'_1, \dotsc, \sigma'_k} \Pr_{r_1, \dotsc, r_k} \left[\bigwedge_{i = 1}^k d_i(e) \in O^{\bar{\cA}_i}(d_i(S_i), r_i, \sigma'_i)\right]\right]
\]
Since the random bits $r_i$ of each algorithm $\cA_i$ are independent of the random bits of the other algorithms $\{\bar{\cA}_j\}_{i \neq j}$, we get that for every $i = 1,\dotsc, k$ the membership of the element $d_i(e)$ in the set $O^{\bar{\cA}_i}(d_i(S_i), r_i, \sigma'_i)$ is independent of this membership for other values of $i$ (given that $\sigma'_1, \dotsc, \sigma'_k$ and $S_1, \dotsc, S_k$ are all considered deterministic). Plugging this observation into the above lower bound on $\Pr[e \in O]$, we get
\begin{align*}
	\Pr[e \in O]
	\geq{} &
	\E_{S_1, \dotsc, S_k} \left[\min_{\sigma'_1, \dotsc, \sigma'_k} \prod_{i = 1}^k \Pr_{r_i} [d_i(e) \in O^{\bar{\cA}_i}(d_i(S_i), r_i, \sigma'_i)]\right]\\
	={} &
	\E_{S_1, \dotsc, S_k} \left[\prod_{i = 1}^k \min_{\sigma'_i} \Pr_{r_i} [d_i(e) \in O^{\bar{\cA}_i}(d_i(S_i), r_i, \sigma'_i)]\right]\\
	={} &
	\prod_{i = 1}^k \E_{S_i} [\min_{\sigma'_i} \Pr_{r_i} [d_i(e) \in O^{\bar{\cA}_i}(d_i(S_i), r_i, \sigma'_i)]]
	\geq
	\prod_{i = 1}^k c^o_i
	\enspace,
\end{align*}
where the last equality holds since $\min_{\sigma'_i} \Pr_{r_i} [d_i(e) \in O^{\bar{\cA}_i}(d_i(S_i), r_i, \sigma'_i)]$ is a function of the sample $S_i$, which is independent of the sample $S_j$ for every $j \neq i$; and the last inequality follows from the facts that $\bar{\cA}_i$ is $(c^o_i, c^a_i)$-weakly-$\optc$-competitive and $d_i(e) \in \OPT'_i \setminus L(\OPT'_i)$. Using the linearity of expectation, we now get
\[
	\E[w'(O)]
	\geq
	\prod_{i = 1}^k c^o_i \cdot w'\mathopen{}\left(\bigcap_{i = 1}^k g^{\cA_i}(\OPT'_i \setminus L(\OPT'_i))\right)
	\enspace.
\]

It remains to show that
\begin{equation} \label{eq:halving}
	w'\mathopen{}\left(\bigcap_{i = 1}^k g^{\cA_i}(\OPT'_i \setminus L(\OPT'_i))\right)
	\geq
	\frac{1}{2} \cdot w'\mathopen{}\left(\bigcap_{i = 1}^k g^{\cA_i}(\OPT'_i)\right)
	\enspace.
\end{equation}
If $\bigcap_{i = 1}^k g^{\cA_i}(\OPT'_i) = \varnothing$, then we are done. Otherwise, for every $i = 1, \dotsc, k$, let $h_i = g^{\cA_i}(h(\OPT'_i))$, and let $i^*$ be one of the values maximizing $w'(h_{i^*})$. Since $w'(h_{i^*}) > 0$, we get that $h_{i^*} \not \in S$ and that $d_i(h_{i^*})$ is not a loop of $M^{\cA_i}$ for any $i = 1, \dotsc, k$. One corollary of these observations is that $d_i(h_{i^*})$ is one of the candidates to be the first element of $\OPT'_i$ when this optimal set is constructed by the greedy algorithm. Since the greedy algorithm picks $h(\OPT'_i)$, we get that $w'(h_i) = w'(h(\OPT'_i)) \geq w'(h_{i^*})$, which can coexist with the definition of $i^*$ only when $h_i = h_{i^*}$. In other words, we have proved that $h_{i^*} = g^{\cA_i}(h(\OPT'_i))$ for every $i = 1, \dotsc, k$. Hence, $\{h_{i^*}\} = \bigcap_{i = 1}^k \{g^{\cA_i}(h(\OPT'_i))\} \subseteq \bigcap_{i = 1}^k g^{\cA_i}(\OPT'_i \setminus L(\OPT'_i))$, where the inclusion holds because $h(\OPT'_i)$ cannot be in $L(\OPT'_i)$. This already proves Inequality~\eqref{eq:halving} when $w'\mathopen{}\left(\bigcap_{i = 1}^k g^{\cA_i}(\OPT'_i)\right) \leq 2w'(h_{i^*})$. Hence, we may assume from now on $w'\mathopen{}\left(\bigcap_{i = 1}^k g^{\cA_i}(\OPT'_i)\right) > 2w'(h_{i^*})$.

For ease of the reading, let us denote $C = \bigcap_{i = 1}^k g^{\cA_i}(\OPT'_i)$. One can observe that the definition of $L(\OPT'_i)$ implies that $g^{\cA_i}(L(\OPT'_i)) \cap C$ is either empty or contains all the elements of $C$ whose weight according to $w'$ is below some threshold. Thus, there must exist some $i'$ such that $g^{\cA_i}(L(\OPT'_i)) \cap C \subseteq g^{\cA_{i'}}(L(\OPT'_{i'})) \cap C$ for every $i = 1, \dotsc, k$. Using this observation, we get
\begin{align*}
	w'\mathopen{}\left(\bigcap_{i = 1}^k g^{\cA_i}(\OPT'_i \setminus L(\OPT'_i))\right)
	={}\mspace{-96mu} & \mspace{96mu}
	w'(C) - w'\mathopen{}\left(C \cap \bigcup_{i = 1}^k g^{\cA_i}(L(\OPT'_i))\right)\\
	={} &
	w'(C) - w'(C \cap g^{\cA_{i'}}(L(\OPT'_{i'})))
	\geq
	w'(C) - w'(g^{\cA_{i'}}(L(\OPT'_{i'})))\\
	={} &
	w'(C) - w'(L(\OPT'_{i'}))
	>
	w'(C) - w'(h(\OPT'_{i'}))
	=
	w'(C) - w'(h_{i^*})
	\enspace,
\end{align*}
where the last inequality holds by the definition of $L(\OPT'_{i'})$. Note that this completes the proof of Inequality~\eqref{eq:halving} since we assume $w'(C) > 2w'(h_{i^*})$.
\end{proof}

It turns out that the weights of $O$ and $I$ can be related.
\begin{lemma} \label{lem:O_I_relation}
$w'(I) \geq w'(O) / \max\{\sum_{i = 1}^k c^a_i, 1\}$.
\end{lemma}
\begin{proof}
Fix $i \in \{1, \dotsc, k\}$. The fact that $\bar{\cA}_i$ is $(c^o_i, c^a_i)$-weakly-$\optc$-competitive implies that
\[
	w'((T^{\bar{\cA}_i}_1 \cup T^{\bar{\cA}_i}_2) \cap O^{\bar{\cA}_i}(d_i(S_i), r_i, \sigma_i)) + c^a_i \cdot w'(T^{\bar{\cA}_i}_1 \setminus O^{\bar{\cA}_i}(d_i(S_i), r_i, \sigma_i)) \geq w'(O^{\bar{\cA}_i}(d_i(S_i), r_i, \sigma_i))
	\enspace.
\]
Recall that $g^{\cA_i}(T^{\bar{\cA}_i}_1) = I$, which implies $T^{\bar{\cA}_i}_1 = d_i(I)$. Thus,
\begin{align} \label{eq:elements_outside_opt}
	c^a_i \cdot w'(I \setminus O)
	={} &
	c^a_i \cdot w'(T^{\bar{\cA}_i}_1 \setminus d_i(O))
	\geq
	c^a_i \cdot w'(T^{\bar{\cA}_i}_1 \setminus O^{\bar{\cA}_i}(d_i(S_i), r_i, \sigma_i))\\ \nonumber
	\geq{} &
	w'(O^{\bar{\cA}_i}(d_i(S_i), r_i, \sigma_i) \setminus (T^{\bar{\cA}_i}_1 \cup T^{\bar{\cA}_i}_2))\\ \nonumber
	\geq{} &
	w'(d_i(O) \setminus (T^{\bar{\cA}_i}_1 \cup T^{\bar{\cA}_i}_2))
	=
	w'(O \setminus g^{\cA_i}(T^{\bar{\cA}_i}_1 \cup T^{\bar{\cA}_i}_2))
	\enspace,
\end{align}
where the first and last inequalities hold since $O \subseteq g^{\cA_i}(O^{\bar{\cA}_i}(d_i(S_i), r_i, \sigma_i))$---which implies the inclusion $d_i(O) \subseteq O^{\bar{\cA}_i}(d_i(S_i), r_i, \sigma_i)$.

We now observe that if an element $e \in O$ belongs to the set $g^{\cA_i}(T^{\bar{\cA}_i}_1 \cup T^{\bar{\cA}_i}_2)$ for every $i = 1,\dotsc, k$, then this means that each one of the algorithms $\bar{\cA}_i$ accepted its corresponding element $d_i(e)$, which implies that $e$ is also a member of $I$. Thus, we can lower bound $w'(I)$ by
\begin{align} \label{eq:I_O_relation}
	w'(I)
	\geq{} &
	w'(I \setminus O) + w'(O) - \sum_{i = 1}^k w'(O \setminus g^{\cA_i}(T^{\bar{\cA}_i}_1 \cup T^{\bar{\cA}_i}_2))\\\nonumber
	\geq{} &
	w'(I \setminus O) + w'(O) - \sum_{i = 1}^k c^a_i \cdot w'(I \setminus O)
	\enspace,
\end{align}
where the second inequality follows from Inequality~\eqref{eq:elements_outside_opt}. Rearranging Inequality~\eqref{eq:I_O_relation}, we get
\[
	w'(I \cap O) + w'(I \setminus O) \cdot \sum_{i = 1}^k c^a_i
	\geq
	w'(O)
	\enspace.
\]
The lemma now follows by observing that the left hand side of the last inequality is upper bounded by $\max\{\sum_{i = 1}^k c^a_i, 1\} \cdot w'(I)$.
\end{proof}

We are now ready to prove Theorem~\ref{thm:main}.

\begin{proof}[Proof of Theorem~\ref{thm:main}]
Recall that the output $I$ of our algorithm is independent in $M_1 \cap \dotso \cap M_k$ by the above discussion. Additionally,
\begin{align*}
	\E[w(I)]
	\geq{} &
	\E[w'(I)]
	\geq
	\frac{\E[w'(O)]}{\max\{\sum_{i = 1}^k c^a_i, 1\}}
	\geq
	\frac{\frac{1}{2}\prod_{i = 1}^k c^o_i \cdot \E\left[w'\mathopen{}\left(\bigcap_{i = 1}^k g^{\cA_i}(\OPT'_i)\right)\right]}{\max\{\sum_{i = 1}^k c^a_i, 1\}}\\
	\geq{} &
	\frac{\frac{1}{2}\prod_{i = 1}^k c^o_i \cdot \E[w(R)]}{4k(k + 1) \cdot \max\{\sum_{i = 1}^k c^a_i, 1\}}
	\geq
	\frac{\prod_{i = 1}^k c^r_i c^o_i \cdot w(\OPT)}{8k(k + 1) \cdot \max\{\sum_{i = 1}^k c^a_i, 1\}}
	\enspace,
\end{align*}
where the first inequality holds since $w(e) \geq w'(e)$ for every element $e \in N$, the second inequality holds due to Lemma~\ref{lem:O_I_relation}, the third due to Lemma~\ref{lem:O_value}, the fourth due to Lemma~\ref{lem:intersection_OPT_general} and the last due to Observation~\ref{obs:R_value}.
\end{proof}

\subsection*{Proof of Lemma~\ref{lem:intersection_OPT_general}}

We begin the proof of the lemma by presenting the algorithm whose existence it guarantees. This algorithm assumes that, for every matroid $M^{\cA_i}$ and element $e \in N$, there exists a fixed order over the refinements $(g^{\cA_i})^{-1}(e)$ of $e$. This order should be independent of $S$. However, if there is no natural candidate for this order, one can simply use a uniformly random order. The first component of the guaranteed algorithm is the generalization of the greedy algorithm given as Algorithm~\ref{alg:generalized_greedy}. This generalization uses $\cI_i$ to denote the collection of independent sets of $M^{\cA_i}$.

\begin{algorithm2e}
  \DontPrintSemicolon
	\lFor{$i = 1 \dotsc, k$}{Let $I_i = \varnothing$.}
  \For{each element $e\in S$ in a decreasing $w$ weight order}
	{
		\If{for every $i = 1,\dotsc, k$ there exists $e_i \in (g^{\cA_i})^{-1}(e)$ such that $I_i + e_i \in \cI_i$}
		{
			\lFor{$i = 1 \dotsc, k$}{Add to $I_i$ the first such element $e_i \in (g^{\cA_i})^{-1}(e)$.}
		}
  }
	\Return{$I_1, \dotsc, I_k$}.
\caption{GeneralizedGreedy$(S)$}\label{alg:generalized_greedy}
\end{algorithm2e}

Like in the case of the standard greedy algorithm, we define an element $e \not \in S$ as greedy-relevant with respect to $S$ if and only if $\generalizedGreedy(S) \neq \generalizedGreedy(S + e)$. Recall that $\imp(S)$ is the set of elements that are greedy-relevant with respect to $S$. For a greedy-relevant element $e$ we also denote by $\role_i(S, e)$ the element of $(g^{\cA_i})^{-1}(e)$ which is added to $I_i$ by $\generalizedGreedy(S + e)$. Using this notation we can now give, as Algorithm~\ref{alg:reduction_general}, the promised algorithm whose existence is guaranteed by Lemma~\ref{lem:intersection_OPT_general}.

\begin{algorithm2e}
  \SetKwInOut{Input}{input}\SetKwInOut{Output}{output}
  \DontPrintSemicolon
  \Input{subset $S\subseteq N$}
  \Output{weight function $\wh$ and functions $d_i\colon N \to N^{\cA_i}$ for every $i = 1,\dotsc, k$}
  \BlankLine
  For each element $e \in S$, define $\wh(e) =0$ and let $d_i(e)$ be the first element of $(g^{\cA_i})^{-1}(e)$.\\
  \For{each other element $e\in N\setminus S$}
	{
		\If{$e \in \imp(S)$}{
			Define $\wh(e) = w(e)$, and let $d_i(e) = \role_i(S, e)$ for every $i = 1, \dotsc, k$.
		}
		\Else
		{
			Define $\wh(e) =0$, and let $d_i(e)$ be the first element of $(g^{\cA_i})^{-1}(e)$.
		}
  }
\caption{GeneralizedOverlappingOPT$(S)$}\label{alg:reduction_general}
\end{algorithm2e}

The following observation can be immediately deduced from Algorithm~\ref{alg:reduction_general}. This observation corresponds to Observation~\ref{obs:weight} from Section~\ref{sec:makeOptsOverlap}.
\begin{observation}
  \label{obs:immediate_results_general}
  On input $S\subseteq N$, Algorithm~\ref{alg:reduction_general} defines a weight
  function $\wh$ and functions $d_i\colon N \to N^{\cA_i}$ for every $i = 1,\dotsc, k$ satisfying:
	\begin{compactitem}
		\item $d_i$ maps every element of $N$ to one of its refinements in $N^{\cA_i}$.
		\item $\wh\leq w$.
		\item $\wh(e) = 0$ whenever $e\in S$ or $d_i(e)$ is a loop in $M^{\cA_i}$ for some $i = 1,\dotsc, k$.
	\end{compactitem}
\end{observation}

To prove Lemma~\ref{lem:intersection_OPT_general} it remains to show that 
\[\E\left[w'\mathopen{}\left(\bigcap_{i = 1}^k g^{\cA_i}(\OPT'_i)\right)\right] \geq \frac{w(R)}{4k(k + 1)} \enspace,\]
where the expectation is over the randomness of $S$, and $\OPT'_i$ is the optimal set in the matroid $M^{\cA_i}|_{d_i(N)}$ with respect to a weight function assigning to each element $e$ of $M^{\cA_i}|_{d_i(N)}$ the weight $w'(g^{\cA_i}(e))$. For that purpose we need to describe an equivalent way to get the functions defined by Algorithm~\ref{alg:reduction_general}. Consider the following extended ground set.
\[
	N' = \bigcup_{e \in N} N'(e)
	\enspace,
	\qquad \text{where} \quad
	N'(e) = \bigtimes_{i = 1}^k (g^{\cA_i})^{-1}(e)
	\enspace.
\]
For every element $e \in N'$ we define $\role_i(e)$ as the element of $N^{\cA_i}$ in $e$. We now extend each matroid $M^{\cA_i}$ to the ground set $N'$ by treating every element $e \in N'$ as equivalent to $\role_i(e)$. More formally, for every matroid $M^{\cA_i}$ we define $\hat{M}^{\cA_i}$ as a matroid over $N'$ such that a set $I \subseteq N'$ is independent in $\hat{M}^{\cA_i}$ if and only if the set $\{\role_i(e) \mid e \in I\}$ is independent in $M^{\cA_i}$ and, additionally, there are no two distinct elements $e_1, e_2 \in I$ such that $\role_i(e_1) = \role_i(e_2)$. Note that the second condition is a technical condition necessary for guaranteeing that $\hat{M}^{\cA_i}$ is indeed a matroid because it ensures that applying the $\role_i$ function to every element of an independent set of $\hat{M}^{\cA_i}$ maps the set to an independent set of $M^{\cA_i}$ of the same size. We also need an additional partition matroid $M^P$ over $N'$ in which a set $I \subseteq N'$ is independent if and only if $|I \cap N'(e)| \leq 1$ for every element $e \in N$.

We extend the weights function $w$ to $N'$ in the natural way, i.e., the weight $w(e')$ of an element $e' \in N(e)$ is defined by $w(e') = w(e)$. In the rest of this section, whenever we refer to the greedy algorithm or to greedy-relevant elements, we assume that the greedy algorithm is executed with respect to the weight function $w$ and the constraint defined by the intersection of the matroids $\hat{M}^{\cA_1}, \dotsc, \hat{M}^{\cA_k}$ and $M^P$. Additionally, we assume that the greedy algorithm uses a tie breaking rule which is consistent with the orders mentioned above within the refinements $(g^{\cA_i})^{-1}(e)$ of every element $e \in N$. More specifically, consider two distinct arbitrary elements $e_1$ and $e_2$ of $N'$ such that $\role_i(e_1)$ and $\role_i(e_2)$ are both refinements of the same element $e \in N$ for every $i = 1, \dotsc, k$. Then, we assume that the greedy algorithm breaks ties between $e_1$ and $e_2$ in favor of $e_1$ whenever it holds that, for every $i = 1, \dotsc, k$, $\role_i(e_1)$ is either equal to $\role_i(e_2)$ or appears before $\role_i(e_2)$. Given these assumptions, one can observe that Algorithm~\ref{alg:reduction_large_groundset} defines the same functions $w'$ and $d_i$ as Algorithm~\ref{alg:reduction_general}.

\begin{algorithm2e}[ht]
  \SetKwInOut{Input}{input}\SetKwInOut{Output}{output}
  \DontPrintSemicolon
  \Input{subset $S\subseteq N$}
  \Output{weight function $\wh$ and functions $d_i\colon N \to N^{\cA_i}$ for every $i = 1,\dotsc, k$}
  \BlankLine
  For each element $e \in S$, define $\wh(e) =0$ and let $d_i(e)$ be the first element of $(g^{\cA_i})^{-1}(e)$.\\
	Let $S' = \bigcup_{e \in S} N'(e)$.\\
  \For{each other element $e\in N\setminus S$}
	{
		\If{$N'(e) \cap \imp(S') \neq \varnothing$}{
			Let $e'$ be the first element in the intersection $N'(e) \cap \imp(S')$.\\
			Define $\wh(e) = w(e)$, and let $d_i(e) = \role_i(e')$ for every $i = 1, \dotsc, k$.
		}
		\Else
		{
			Define $\wh(e) =0$, and let $d_i(e)$ be the first element of $(g^{\cA_i})^{-1}(e)$.
		}
  }
\caption{GeneralizedOverlappingOPT$(S)$}\label{alg:reduction_large_groundset}
\end{algorithm2e}

Recalling that $S$ contains every element of $N$ with probability $p$, one can observe that Algorithm~\ref{alg:simp_general} is an offline algorithm which constructs sets $G$ and $W$ with a joint distribution identical to the random variables $G=\greedy(\bigcup_{e \in S} N'(e))$ and $W=\bigcap_{j=1}^k g^{\cA_i}(\OPTh_j)$. This algorithm uses $\hat{\cI}_i$ to denote the collection of independent sets of $\hat{M}^{\cA_i}$.
   \begin{algorithm2e}[ht]
        \DontPrintSemicolon
        \SetKwFor{Prob}{With probability}{}{endprob}
        \SetKwFor{Otherwise}{Otherwise}{}{endother}
        $G= \varnothing$, $W = \varnothing$ and $\OPTh_j = \varnothing$ \quad $\forall j\in [k]$\\
								\For{every element $e \in N$ in a decreasing $w$ weight order}{
          \If{there is an element $e' \in N(e)$ such that $G+e' \in \bigcap_{i = 1}^k \hat{\cI}_i$}{
						Let $e'$ be the first such element.\\
            \Prob{$p$:}{
              $G= G+e'$\;
            }
            \Otherwise{ (with remaining probability $1-p$):}{
              \For{$i=1, \ldots, k$} { 
                \lIf{$\OPTh_i + \role_i(e') \in \cI_i$}{$\OPTh_i = \OPTh_i + \role_i(e')$}
              } 
              \lIf{$e \in \cap_{i=1}^k g^{\cA_i}(\OPTh_i) $}{$W= W+e$}
            }
            
          }
        }
      \caption{$\simp(p)$}\label{alg:simp_general}
   \end{algorithm2e}

At this point we observe that Algorithm~\ref{alg:simp_general} is very similar to Algorithm~\ref{alg:simp}. In particular, it is similar enough for the proof from Section~\ref{sec:makeOptsOverlap} of the inequality
\begin{equation} \label{eq:inherited}
	\E[w'(W)]
	\geq
	\frac{\E[w(G)]}{2(2k-1)}
\end{equation}
to carry over, with only technical modifications, also to Algorithm~\ref{alg:simp_general}. Additionally, let $R'$ be a subset of $N'$ obtained by picking for every element $e \in R$ an element $e' \in N'(e)$ such that $\role_i(e') \in R^{\cA_i}(\OPT)$ for every $i = 1,\dotsc, k$. Observe that $w(R') = w(R)$, and $R'$ is independent in $\hat{M}^{\cA_i}$ for every $i = 1,...,k$ since $R^{\cA_i}(\OPT)$ is independent in $M^{\cA_i}$ ($R'$ is also independent in $M^P$ by construction). Moreover, since $G$ has the same distribution as the output of the greedy algorithm on $\bigcup_{e \in S} N'(e)$, and the greedy algorithm is a $(k + 1)$-approximation algorithm when its constraint is the intersection of $k + 1$ matroids, we get
\[
	\E[w(G)]
	\geq
	\E\left[\frac{w(R' \cap \bigcup_{e \in S} N'(e))}{k + 1}\right]
	=
	\frac{\E[w(R' \cap \bigcup_{e \in S} N'(e))]}{k + 1}
	=
	\frac{p \cdot w(R')}{k + 1}
	=
	\frac{(2k - 1) \cdot w(R)}{2k(k + 1)}
	\enspace.
\]
where the expectation is over the choice of $S$. Combining this inequality with Inequality~\eqref{eq:inherited}, we get
\[
	\E\left[w'\mathopen{}\left(\bigcap_{i = 1}^k g^{\cA_i}(\OPT'_i)\right)\right]
	=
	\E[w'(W)]
	\geq
	\frac{\E[w(G)]}{2(2k-1)}
	\geq
	\frac{\E[w(R)]}{4k(k + 1)}
	\enspace,
\]
which completes the proof of Lemma~\ref{lem:intersection_OPT_general}.

\bibliographystyle{plain}
\bibliography{lit}

\appendix

\section{Algorithms for specific classes of matroids}\label{sec:algForSpecMats}

In this appendix we describe algorithms for {\MSP} on various specific classes of matroids. All the algorithms we give here can be fitted into our framework. Moreover, the competitive ratio of the algorithm we describe for every class matches up to a (small) constant the state of the art competitive ratio for {\MSP} on this class of matroids.

\subsection{Partition matroids}

In a (generalized) partition matroid the ground set $N$ is the disjoint union of multiple partitions $N_1, \dotsc, N_\rho$, and each partition is associated with a positive integer parameter $\ell_1, \dotsc, \ell_\rho$. A set $I$ is independent in such a matroid if and only if $|I \cap N_i| \leq \ell_i$ for every $i = 1, \dotsc, \rho$. A particularly interesting class of partition matroids is the class of simple partition matroids, in which the parameters $\ell_1, \dotsc, \ell_\rho$ are all equal to $1$ (i.e., an independent set may include at most one element from each partition of the ground set). Azar et al.~\cite{azar_2014_prophet} give an order-oblivious $4$-competitive secretary algorithm on rank $1$ matroids, and this algorithm can be easily extended to simple partition matroids since the selection of an element from each partition of a simple partition matroid can be done independently. Nevertheless, for the sake of completeness we give here explicitly, as Algorithm~\ref{alg:partition}, the algorithm obtained for simple partition matroids. This algorithm has a parameter $p$, which is a probability to be determined later. Additionally, the algorithm uses the notation $N(e)$ to denote the partition of element $e$.

\begin{algorithm2e}
\caption{Order-Oblivious {\MSP} Algorithm on Simple Partition Matroids} \label{alg:partition}
Pick a random value $m$ from the binomial distribution $B(|N|, p)$.\\
Request a sample $S$ of size $m$.\\
Let $I \gets \varnothing$.\\
\For{every arriving non-sample element $e$}{
	Let $t_e$ be the maximum weight element in $S \cap N(e)$.\\
	\If{$I \cap N(e) = \varnothing$ and ($S \cap N(e) = \varnothing$ or $w(e) > w(t_e)$)}{
		Add $e$ to $I$.
	}
}
\Return{$I$}.
\end{algorithm2e}

One can observe that Algorithm~\ref{alg:partition} is an order-oblivious algorithm which always outputs a solution set $I$ which is independent in the simple partition matroid, and moreover, it can be implemented using only independence oracle queries on sets of elements that already arrived. Thus, we concentrate here on analyzing the competitive properties of Algorithm~\ref{alg:partition}.

\begin{theorem} \label{thm:partition_algorithm}
Algorithm~\ref{alg:partition} is a $p(1 - p)$-$\optc$-competitive order-oblivious algorithm for {\MSP} on simple partition matroids. Hence, for $p = \nicefrac{1}{2}$ it is $\nicefrac{1}{4}$-$\optc$-competitive.
\end{theorem}
\begin{proof}
We begin the proof by noting that the way $S$ is defined as a uniformly random subset of $N$ of size $m$, where $m$ is distributed according to the binomial distribution $B(|N|, p)$, implies that $S$ contains every element of $N$ with probability $p$, independently.

Consider now an arbitrary element $e^* \in \OPT$, where $\OPT$ is the maximum weight independent set of the simple partition matroid. Note that the weight of $e^*$ is the largest weight of any element in $N(e^*)$ because otherwise we could improve $\OPT$ by replacing $e^*$ with a heavier element of its partition. There are now two cases to consider. We begin with the case that $N(e^*)$ contains other elements beside $e^*$, and let $e_2$ be the heaviest element in $N(e^*) - e^*$. Observe that in this case Algorithm~\ref{alg:partition} is guaranteed to take $e^*$ when $e_2 \in S$ and $e^* \not \in S$ because the presence of $e_2$ in $S$ prevents any element of $N(e^*)$ other than $e^*$ itself from being selected by the algorithm. Moreover, this happens with probability $p(1 - p)$ by the above discussion regarding the distribution of $S$. It remains to consider the case that $N(e^*)$ contains only $e^*$. In this case Algorithm~\ref{alg:partition} selects $e^*$ whenever $e^* \not \in S$, which happens with probability $1 - p \geq p(1 - p)$.
\end{proof}

The above theorem implies, by the discussion in Section~\ref{sec:genToMatInt}, the following corollary.

\begin{corollary}
For $p = \nicefrac{1}{2}$, Algorithm~\ref{alg:partition} is a $(\nicefrac{1}{4}, 0)$-weakly-$\optc$-competitive order-oblivious algorithm for {\MSP} on simple partition matroids. Moreover, it implies a $(1, \nicefrac{1}{4}, 0)$-reduce-and-solve algorithm for this problem.
\end{corollary}

We now shift our attention to the more general class of (generalized) partition matroids. The {\MSP} algorithm we give for this class is Algorithm~\ref{alg:generalized_partition}. This algorithm denotes by $N(e)$ the partition of element $e$, and by $\ell(e)$ the parameter associated with this partition.

\begin{algorithm2e}
\caption{Order-Oblivious {\MSP} Algorithm on Partition Matroids} \label{alg:generalized_partition}
Pick a random value $m$ from the binomial distribution $B(|N|, \nicefrac{1}{2})$.\\
Request a sample $S$ of size $m$.\\
Let $I \gets \varnothing$.\\
\For{every arriving non-sample element $e$}{
	Let $t_e$ be the element of $S \cap N(e)$ having the $\ell(e)$-th largest weight in this set.\\
	\If{$|I \cap N(e)| < \ell(e)$ and ($|S \cap N(e)| < \ell(e)$ or $w(e) > w(t_e)$) \label{line:if}}{
		Add $e$ to $I$.
	}
}
\Return{$I$}.
\end{algorithm2e}

Once again, one can observe that Algorithm~\ref{alg:partition} is an order-oblivious algorithm which always outputs a solution set $I$ which is independent in the partition matroid. The following observation shows that, additionally, Algorithm~\ref{alg:partition} can be implemented using only independence oracle queries on sets of elements that already arrived.
\begin{observation}
Algorithm~\ref{alg:generalized_partition} can be implemented using only independence oracle queries on sets of elements that already arrived.
\end{observation}
\begin{proof}
Let us consider the arrival of an arbitrary element $e$ during the second (non-sample) phase of Algorithm~\ref{alg:generalized_partition}. This arrival requires the algorithm to check a few conditions. One of these conditions is whether $|I \cap N(e)| < \ell(e)$, which can be checked by simply checking whether $I + e$ is independent since $I$ itself is always independent. Implementing the other conditions is somewhat more involved. Let $\OPT(S)$ be the maximum weight independent set in $S$. Notice that $\OPT(S)$ can be found using the greedy algorithm. If $\OPT(S) + e$ is independent, then we get that $|S \cap N(e)| < \ell(e)$, which allows us to implement the other conditions of the algorithm (notice that $t_e$ is not well-defined in this case, which is fine because in this case we do not need it for evaluating the condition of the if statement on Line~\ref{line:if}). Consider now the case that $\OPT(S) + e$ is not independent. In this case we have $|S \cap N(e)| \geq \ell(e)$. Moreover, $t_e$ is the lightest element of $\OPT(S) + e$, other than $e$ itself, whose removal restores the independence of $\OPT(S) + e$.
\end{proof}

It remains to analyze the competitive properties of Algorithm~\ref{alg:generalized_partition}.

\begin{theorem}
Algorithm~\ref{alg:generalized_partition} is a $\nicefrac{1}{4}$-$\optc$-competitive order-oblivious algorithm for {\MSP} on partition matroids.
\end{theorem}
\begin{proof}
Like in the proof of Theorem~\ref{thm:partition_algorithm}, we have here also that $S$ contains every element of $N$ with probability $\nicefrac{1}{2}$, independently. Consider now an arbitrary element $e^* \in \OPT$, where $\OPT$ is the maximum weight independent set of the partition matroid. Note that $\OPT \cap N(e^*)$ contains the $\ell(e^*)$ elements of $N(e^*)$ with the largest weights (unless there are less than $\ell(e^*)$ elements in $N(e^*)$ of non-zero weight, in which case $\OPT \cap N(e^*)$ simply consists of all these elements) because otherwise we could improve $\OPT$ by replacing one of its elements in $N(e^*)$ with a heavier element from this partition.

Let us define now two events. The event $\cE_1$ is the event that $e^* \not \in S$, and the event $\cE_2$ is the event that one of the following happens:
\begin{itemize}
	\item $w(e^*) > w(t_{e^*})$, and there are less than $\ell(e^*)$ elements in $(N(e^*) - e^*) \setminus S$ whose weight is larger than $w(t_{e^*})$.
	\item $t_{e^*}$ is not defined since $|S \cap N(e^*)| < \ell(e^*)$, and $|(N(e^*) - e^*) \setminus S| < \ell(e^*)$.
\end{itemize}
One can observe that the events $\cE_1$ and $\cE_2$ imply together that $e^*$ is selected by Algorithm~\ref{alg:generalized_partition}. Thus, it is enough to prove that $\Pr[\cE_1] \cdot \Pr[\cE_2 \mid \cE_1] \geq 1/4$. Moreover, since it clearly holds that $\Pr[\cE_1] = 1/2$, it remains to show that $\Pr[\cE_2 \mid \cE_1] \geq 1/2$.

We prove the last inequality by showing that, conditioned on $\cE_1$, if $\cE_2$ does not happen when $S - e^*$ is equal to some set $A$, then it must happen when $S - e^* = (N - e^*) \setminus A$; which implies the inequality since it implies that, given $\cE_1$, $\cE_2$ holds for at least one half of the possible values for $S - e^*$ (and each of these values happens with equal probability). Hence, in the rest of the proof we assume that $\cE_1$ happens, and concentrate on a set $A$ such that $\cE_2$ does not happen when $S - e^* = A$. We first observe that whenever $t_{e^*}$ is well-defined (i.e., $|S \cap N(e^*)| \geq \ell(e^*)$), then the fact that $\cE_1$ happened (and thus $e^*$---which is an element of $\OPT \cap N(e^*)$---does not belong to $S$) implies that $t_{e^*}$ is not an element of $\OPT$ since it is defined as the element with the $\ell(e^*)$-th largest weight in $S \cap N(e^*)$. In particular, we get that the inequality $w(e^*) > w(t_{e^*})$ holds whenever $t_{e^*}$ is well defined. We now need to consider two cases.
\begin{itemize}
	\item If $|A \cap N(e^*)| < \ell(e^*)$, then the fact that $\cE_2$ did not happen implies that $|(N - e^*) \setminus A| \geq \ell(e^*)$. This means that when $S - e^* = (N - e^*) \setminus A$ then $t_{e^*}$ is well defined, and thus, $w(e^*) \geq w(t_{e^*})$ Additionally, we have in this case $(N(e^*) - e^*) \setminus S = (N(e^*) - e^*) \setminus [(N - e^*) \setminus A] = A \cap N(e^*)$, which implies that the total number of elements in $(N(e^*) - e^*) \setminus S$ is less than $\ell(e^*)$, and thus, the number of elements in this set whose weight is larger than $w(t_{e^*})$ must also be less than $\ell(e^*)$. Combining all the above, we get that $\cE_2$ happens when $S - e^* = (N - e^*) \setminus A$.
	\item It remains to consider the case that $|A \cap N(e^*)| \geq \ell(e^*)$. In this case $t_{e^*}$ is well defined when $S - e^* = A$, and thus, we have $w(e^*) > w(t_{e^*})$ by the above discussion. Since $\cE_2$ does not happen when $S - e^* = A$, this implies that $(N(e^*) - e^*) \setminus A$ contains at least $\ell(e^*)$ elements whose weight is larger than $w(t_1)$, where $t_1$ is the element which becomes $t_{e^*}$ when $S - e^* = A$. If we now denote by $t_2$ the element that becomes $t_{e^*}$ when $S = (N - e^*) \setminus A$, then the above observation implies $w(t_2) > w(t_1)$. Moreover, by definition, there are less than $\ell(e^*)$ elements in $A \cap N(e^*)$ whose weight is larger than $w(t_1)$, and thus, when $S = (N - e^*) \setminus A$ there are less than $\ell(e^*)$ elements in $(N(e^*) - e^*) \setminus S = A \cap N(e^*)$ whose weight is larger than $w(t_{e^*}) = w(t_2) > w(t_1)$. Hence, we get that $\cE_2$ happens when $S - e^* = (N - e^*) \setminus A$. \qedhere
\end{itemize}
\end{proof}

The above theorem implies, by the discussion in Section~\ref{sec:genToMatInt}, the following corollary.

\begin{corollary}
Algorithm~\ref{alg:generalized_partition} is a $(\nicefrac{1}{4}, 0)$-weakly-$\optc$-competitive order-oblivious algorithm for {\MSP} on partition matroids. Moreover, it implies a $(1, \nicefrac{1}{4}, 0)$-reduce-and-solve algorithm for this problem.
\end{corollary}

\subsection{Graphic matroids}

A graphic matroid is defined by a graph $G = (V, N)$ whose edges form the ground set of the matroid. A set of edges is independent in the graphic matroid if and only if it is acyclic. Azar et al.~\cite{azar_2014_prophet} explain how to get a $\nicefrac{1}{8}$-competitive order-oblivious algorithm for the secretary problem on graphic matroids based on an algorithm of~\cite{korula_2009_algorithms}. Unfortunately, this algorithm is not known to be $\optc$-competitive (or even weakly-$\optc$-competitive) for any constant. However, it can be easily adapted to be a $(\nicefrac{1}{2}, \nicefrac{1}{4}, 0)$-reduce-and-solve algorithm. Here we present a different reduce-and-solve algorithm for this problem which has different values for the parameters. It is important to note that both the algorithms of~\cite{azar_2014_prophet,korula_2009_algorithms} and our algorithm assume more than independence oracle access to the graphic matroid. Specifically, they assume the ability to determine the endpoints for every given edge $e \in N$.

The reduce-and-solve algorithm $\cA$ we suggest for the graphic matroid consists of the following components:
\begin{itemize}
\item The ground set $N^\cA$ consists of two elements for every edge $e \in N$. These elements are denoted by $e^u$ and $e^v$, where $u$ and $v$ are the end points of $e$. Accordingly, the function $g^\cA \colon N^\cA \to N$ assigns $e$ to be the source element for both $e^u$ and $e^v$. Note that it is possible to evaluate $(g^\cA)^{-1}(e) = \{e^u, e^v\}$ given access to $e$ alone.
\item The matroid $M^\cA$ is a simple partition matroid over the ground set $N^\cA$ containing a single partition for every node of $G$. The partition of a node $u \in V$, which we denote by $N^\cA(u)$, contains the elements $\{e^u \mid e \in \delta(u)\}$---where $\delta(u)$ denotes the set of edges of $G$ hitting $u$.
\item The order-oblivious algorithm for {\MSP} on restrictions of $M^\cA$ is given by Algorithm~\ref{alg:graphic_partition}. This algorithm denotes by $N'$ the ground set of the restriction on which it is applied. Additionally, one can note that, aside from notational issues, this algorithm is identical to Algorithm~\ref{alg:partition} except for the change in Line~\ref{line:addition_graphic_partition} preventing an element $e^u$ from being added to $I$ if that addition will result in a cycle in $g^\cA(I)$.
\end{itemize}

\begin{algorithm2e}
\caption{Order-Oblivious {\MSP} Algorithm on Simple Partition Matroids Induced by Graphic Matroids} \label{alg:graphic_partition}
Pick a random value $m$ from the binomial distribution $B(|N'|, p)$.\\
Request a sample $S$ of size $m$.\\
Let $I \gets \varnothing$.\\
\For{every arriving non-sample element $e^u$}{
	Let $t_u$ be the maximum weight element in $S \cap N^\cA(u)$.\\
	\If{$I \cap N^\cA(u) = \varnothing$ and ($S \cap N^\cA(u) = \varnothing$ or $w(e^u) > w(t_u)$)}{
		Add $e^u$ to $I$ if that does not create a cycle in $g^\cA(I)$. \label{line:addition_graphic_partition}
	}
}
\Return{$I$}
\end{algorithm2e}

One can note that the output set $I$ of Algorithm~\ref{alg:graphic_partition} is independent in $M^\cA$, and moreover, $g^\cA(I)$ is independent in the original graphic matroid. Thus, $\cA$ is a valid reduce-and-solve algorithm for graphic matroids. It remains, however, to determine the parameters of this algorithm. The next lemma shows that $c^r = 1$ for this algorithm.

\begin{lemma}
For every independent set $B$ of the graphic matroid, there exists an independent set $R^\cA(B)$ of $M^\cA$ such that $R^\cA(B) \cap (g^\cA)^{-1}(e) \neq \varnothing$ for every edge $e \in B$.
\end{lemma}
\begin{proof}
We assume for the sake of simplicity that the edges of $B$ induce a connected graph. If that is not the case, then the proof has to be applied separately for every connected component of $B$.

Fix a root node $r$ for the graph induced by $B$, and let us denote for every edge $e \in B$ by $d(e)$ the end point of $e$ which is further away from $r$. Note that $d(e)$ is well defined since the independence of $B$ implies that the graph induced by $B$ is a tree. We now define $R^\cA(B) = \{e^{d(e)} \mid e \in B\}$. Clearly, we have by definition $R^\cA(B) \cap (g^\cA)^{-1}(e) \neq \varnothing$ for every edge $e \in B$. Additionally, $R^\cA(B)$ must be independent in $M^\cA$ because otherwise there must be two distinct edges $e_1, e_2 \in B$ such that $d(e_1) = d(e_2)$, and this leads to a contradiction as it implies that $B$ contains two different paths from $r$ to $d(e_1)$ ($= d(e_2)$).
\end{proof}

The next lemma gives values for the two other parameters of $\cA$.

\begin{lemma}\label{lem:graphicMatExt}
Algorithm~\ref{alg:graphic_partition} is a $p^2(1-p)$-$\optc$-competitive algorithm for {\MSP} on restrictions of $M^\cA$. Hence, for $p = \nicefrac{2}{3}$ it is $\nicefrac{4}{27} (\approx 0.148)$-$\optc$-competitive.
\end{lemma}
\begin{proof}
Consider now an arbitrary element $e^u \in \OPT$, where $\OPT$ is the maximum weight independent set of the input restriction of $M^\cA$. Let $\cE_1$ be the event that Algorithm~\ref{alg:graphic_partition} gets to Line~\ref{line:addition_graphic_partition} with $e^u$ (i.e., Algorithm~\ref{alg:graphic_partition} adds $e^u$ to $I$ unless this creates a cycle in $g^\cA(I)$), and moreover, $e^u$ is the first element of $N^\cA(u)$ with which Algorithm~\ref{alg:graphic_partition} gets to Line~\ref{line:addition_graphic_partition}. Following the same argument as in the proof of Theorem~\ref{thm:partition_algorithm}, we get that the probability of $\cE_1$ is at least $p(1 - p)$. In the rest of the proof we show that there exists an event $\cE_2$ which is independent of $\cE_1$, has a probability of at least $p$ and implies that the addition of $e^u$ to $I$ does not create a cycle in $g^\cA(I)$. Note that the lemma follows from the existence of such an event.

Let $v$ be the end point of $e$ which is not $u$, and let $h_v$ be the heaviest element in $N^\cA(v) \cap N'$. We define $\cE_2$ as the event that $h_v \in S$. If $h_v$ does not exist, i.e., $N^\cA(v) \cap N' = \varnothing$, then we simply assume $\cE_2$ happens always. Note that $\cE_1$ depends only on elements of the partition $N^\cA(u)$, and thus, it is independent of $\cE_2$ which depends on an element of $N^\cA(v)$. It remains to explain why $\cE_2$ implies that the addition of $e^u$ to $I$ does not create a cycle in $g^\cA(I)$. Assume towards a contradiction that this is not the case, and consider the connected component $C$ of $e$ in $g^\cA(I)$ after the addition of $e^u$ to $I$. Algorithm~\ref{alg:graphic_partition} adds to $I$ at most one element from each partition of a node of $C$. Moreover, $\cE_2$ implies that no element from the partition of $v$ is added to $C$, and thus, the number of elements in $I$ from partitions of nodes of $C$ must be strictly less than the number of such nodes. Since every edge in $C$ corresponds to at least one element of $I$ belonging to the partitions of the nodes of $C$, we get also that the number of edges in $C$ is strictly less than the number of nodes in it. However, this is a contradiction since $C$ is, by definition, a connected component containing a cycle.
\end{proof}

\begin{corollary}\label{cor:graphicMat}
For $p = \nicefrac{2}{3}$, $\cA$ is a $(1, \nicefrac{4}{27} \approx 0.148, 0)$-reduce-and-solve algorithm for {\MSP} on graphic matroids.
\end{corollary}

\subsection{Column-sparse linear matroids} \label{ssc:sparse}

A linear matroid is a matroid defined by a matrix $C$ over a field $\mathbb{F}$. The ground set $N$ of this matroid is the set of columns of $C$, and a set of columns is independent in this matroid if and only if the columns of the set are independent as linear vectors. Graphic matroids can be represented as linear matroids defined by a matrix $C$ containing at most two non-zero values in each column. Motivated by this observation, Soto~\cite{soto_2013_matroid} studies {\MSP} on linear matroids defined by a matrix with at most $k$ non-zero values in each column, where $k$ is an arbitrary constant, and describes a $\nicefrac{1}{ke}$-competitive algorithm for this problem. Like the state of the art algorithm for graphic matroids, this algorithm is not order-oblivious, but it can be made order-oblivious at the cost of a small loss in the competitive ratio (specifically, it deteriorates to $\nicefrac{1}{4k}$). Moreover, this order-oblivious version of the algorithm is not known to be $\optc$-competitive (or even weakly-$\optc$-competitive) for any constant, but it can be easily adapted to be a $(\nicefrac{1}{k}, \nicefrac{1}{4}, 0)$-reduce-and-solve algorithm. Here we take a different approach, and adapt the reduce-and-solve algorithm given above for graphic matroids to this more general setting, which results in a reduce-and-solve algorithm with different values for the parameters. It is important to note that both the algorithm of~\cite{soto_2013_matroid} and our algorithm assume more than independence oracle access to the matroid. Specifically, they assume the ability to determine, given a column of the matrix $C$, the rows in which this column has non-zero values.

The reduce-and-solve algorithm $\cA$ we suggest for {\MSP} on linear matroids defined by a matrix $C$ containing at most $k$ non-zero values in each column consist of the following components:
\begin{itemize}
\item The ground set $N^\cA$ is defined as $\{(c, r) \mid c \in N, r \in \nonzero(c)\}$, where $\nonzero(c)$ is the set of rows in which the column $c$ has non-zero values. Additionally, the function $g^\cA \colon N^\cA \to N$ assigns every element $(c, r)$ to its corresponding column $c$. Note that it is possible to evaluate $(g^\cA)^{-1}(c)$ given access to the column $c$ alone.\footnote{Technically, for $\cA$ to be a valid reduce-and-solve algorithm we need to guarantee that $(g^\cA)^{-1}(c)$ is not empty for every $c \in N$. This does not happen when some columns contain only zero values, and thus, we need to create a dummy element $e_c$ for every such column $c$, add it to $N^\cA$, set $(g^\cA)^{-1}(c) = \{e_c\}$ and make $e_c$ a loop element of $M^\cA$. For simplicity, we ignore this technical issue.}
\item The matroid $M^\cA$ is a partition matroid over the ground set $N^\cA$ containing a single partition for every row of $C$. The partition of row $r$, which we denote by $N^\cA(r)$, contains the elements $\{(c, r) \mid r \in \nonzero(c)\}$.
\item The order-oblivious algorithm for {\MSP} on restrictions of $M^\cA$ is given by Algorithm~\ref{alg:sparse_linear_partition}. Like Algorithm~\ref{alg:graphic_partition}, this algorithm denotes by $N'$ the ground set of the restriction on which it is applied. Moreover, one can note that, aside from notational issues, this algorithm is identical to Algorithm~\ref{alg:partition} except for the change in Line~\ref{line:addition_sparse_linear_partition} preventing an element $(c, r)$ from being added to $I$ if $I$ already contains an element from one of the partitions associated with the other rows in which $c$ has a non-zero value.
\end{itemize}

\begin{algorithm2e}
\caption{Order-Oblivious {\MSP} Algorithm on Simple Partition Matroids Induced by Column-Sparse Linear Matroids} \label{alg:sparse_linear_partition}
Pick a random value $m$ from the binomial distribution $B(|N'|, p)$.\\
Request a sample $S$ of size $m$.\\
Let $I \gets \varnothing$.\\
\For{every arriving non-sample element $(c, r)$}{
	Let $t_r$ be the maximum weight element in $S \cap N^\cA(r)$.\\
	\If{$I \cap N^\cA(r) = \varnothing$ and ($S \cap N^\cA(r) = \varnothing$ or $w((c, r)) > w(t_r)$)}{ 
		Add $(c, r)$ to $I$ if $I \cap N^\cA(r') = \varnothing$ for every $r' \in \nonzero(c)$. \label{line:addition_sparse_linear_partition}
	}
}
\Return{$I$}.
\end{algorithm2e}

The following lemma implies that $\cA$ is a valid reduce-and-solve algorithm for {\MSP} on linear matroids defined by a matrix $C$ containing at most $k$ non-zero values in each column.

\begin{lemma}
The set $g^\cA(I)$, where $I$ is the output set of Algorithm~\ref{alg:sparse_linear_partition}, is always independent in the original linear matroid.
\end{lemma}
\begin{proof}
Assume towards a contradiction that this is not the case, and let $C$ be an arbitrary cycle in $g^\cA(I)$. Additionally, let $(c, r)$ be the first element that was added to $I$ such that $g^\cA((c, r)) \in C$. We observe that $C$ must contain a second column $c' \neq c$ containing the row $r$---otherwise, the removal of $g^\cA((c. r)) = c$ from $C$ could not make $C$ an independent set, as required by the definition of $C$ as a cycle. Next, let us denote by $(c', r')$ an element of $I$ such that $g^\cA((c', r')) = c'$. When $(c', r')$ arrived the solution $I$ already contained the element $(c, r)$ which belongs to the partition $N^\cA(r)$. Note that this leads to a contradiction since Algorithm~\ref{alg:sparse_linear_partition} is guaranteed not to add $(c', r')$ to $I$ if there is already an element in $I$ from one of the partitions associated with rows in which $c'$ has a non-zero value, and $r$ is one of these rows.
\end{proof}

It remains to determine the parameters of $\cA$ as a reduce-and-solve algorithm. The next lemma shows that $c^r = 1$ for this algorithm.

\begin{lemma}
For every independent set $B$ of the original linear matroid, there exists an independent set $R^\cA(B)$ of $M^\cA$ such that $R^\cA(B) \cap (g^\cA)^{-1}(c) \neq \varnothing$ for every column $c \in B$.
\end{lemma}
\begin{proof}
Consider a bipartite graph $G$ in which one side consists of the columns of $C$, the other side consists of the rows of $C$ and there is an edge between a given pair of a column $c$ and a row $r$ if and only if $r \in \nonzero(c)$. Soto~\cite{soto_2013_matroid} observed that, since $B$ is independent, there must be a matching $M_B$ in this graph covering all columns of $B$. This follows from Hall's theorem since the independence of $B$ implies that $\bigcup_{c \in B'} \nonzero(c)$ contains at least $|B'|$ rows for every subset $B'$ of $B$. Using the matching $M_B$, we can now define
\[
	R^\cA(B)
	=
	\{(c, r) \mid c \in B \text{ and $c$ is matched to $r$ by $M_B$}\}
	\enspace.
\]
Since this is a matching, $R^\cA(B)$ contains at most a single element from each partition of $M^\cA$, and is, thus, independent in $M^\cA$. Additional, for every column $c \in B$, the set $(g^\cA)^{-1}(c)$ contains all the element of $\{(c, r) \mid r \in \nonzero(c)\}$, and thus, has a non-empty intersection with $R^\cA(B)$.
\end{proof}

The next lemma gives values for the two other parameters of $\cA$.

\begin{lemma}
Algorithm~\ref{alg:sparse_linear_partition} is a $p^k(1-p)$-$\optc$-competitive algorithm for {\MSP} on restrictions of $M^\cA$. Hence, for $p = (k - 1)/k$ it is $k^{-1}(1 - 1/k)^k$-$\optc$-competitive.
\end{lemma}
\begin{proof}
Consider now an arbitrary element $(c, r) \in \OPT$, where $\OPT$ is the maximum weight independent set of the input restriction of $M^\cA$. Let $\cE_1$ be the event that Algorithm~\ref{alg:sparse_linear_partition} gets to Line~\ref{line:addition_sparse_linear_partition} with $(c, r)$ (i.e., Algorithm~\ref{alg:sparse_linear_partition} adds $(c, r)$ to $I$ unless $I$ already contains an element from some partition of $M^\cA$ associated with a different row in which $c$ takes a non-zero value), and moreover, $(c, r)$ is the first element of $N^\cA(r)$ with which Algorithm~\ref{alg:sparse_linear_partition} gets to Line~\ref{line:addition_sparse_linear_partition}. Following the same argument as in the proof of Theorem~\ref{thm:partition_algorithm}, we get that the probability of $\cE_1$ is at least $p(1 - p)$.
In the rest of the proof we show that there exists an event $\cE_2$ which is independent of $\cE_1$, has a probability of at least $p^{k-1}$ and implies that no element is ever added to $I$ from the partitions associated with rows other than $r$ in which $c$ takes a non-zero value. Note that the lemma follows from the existence of such an event.

For every row $r' \in \nonzero(c) - r$, let $h_{r'}$ denote the heaviest element in $N^\cA(r) \cap N'$. We define $\cE_2$ as the event that $h_{r'} \in S$ for every row $r' \in \nonzero(c) - r$ having $N^\cA(r') \cap N' \neq \varnothing$. Note that $\cE_1$ depends only on elements of the partition $N^\cA(r)$, and thus, it is independent of $\cE_2$ which depends on elements from other partitions. Additionally, the probability of $\cE_2$ is at least $p^{k - 1}$ since every element of $N'$ belongs to $S$ with probability $p$, independently. It remains to observe that $\cE_2$ implies that, for every row $r' \in \nonzero(c) - r$ associated with a non-empty partition $N^\cA(r') \cap N'$, the heaviest element of the partition $N^\cA(r') \cap N'$ appears within the sample $S$, and this prevents Algorithm~\ref{alg:sparse_linear_partition} from adding to $I$ any elements of this partition.
\end{proof}

\begin{corollary}
For $p = (k - 1)/k$, $\cA$ is a $(1, k^{-1}(1 - 1/k)^k, 0)$-reduce-and-solve algorithm for {\MSP} on linear matroids defined by a matrix $C$ containing at most $k$ non-zero values in each column.
\end{corollary}

\subsection{Co-graphic matroids}

A co-graphic matroid is defined by a graph $G = (V, N)$ whose edges form the ground set of the matroid. A set of edges is independent in the co-graphic matroid if and only if their removal from $G$ does not increase the number of connected components in $G$. Alternatively, as implied by its name, the co-graphic matroid can also be defined as the dual matroid of the graphic matroid defined by $G$. Soto~\cite{soto_2013_matroid} describes a procedure that given a co-graphic matroid outputs a random simple partition matroid such that every element belonging to the optimal solution in the original co-graphic matroid also belongs to the optimal solution of the new simple partition matroid with probability $\nicefrac{1}{3}$.


Soto~\cite{soto_2013_matroid} used his procedure to get a $\nicefrac{1}{3e}$-competitive algorithm for {\MSP} on co-graphic matroids by applying a $\nicefrac{1}{e}$-competitive {\MSP} algorithm on the simple partition matroid produced by the procedure. Unfortunately, the algorithm obtained this way is not order-oblivious because the $\nicefrac{1}{e}$-competitive algorithm for {\MSP} on simple partition matroids is not order-oblivious (or more accurately, it is not $\nicefrac{1}{e}$-competitive when the arrival order in the second phase is adversarial). This issue can be fixed by using instead of this algorithm the $\nicefrac{1}{4}$-$\optc$-competitive order-oblivious algorithm for {\MSP} on simple partition matroids given by Algorithm~\ref{alg:partition}. This gives us the following result.

\begin{theorem}
There exists a $\nicefrac{1}{12}$-$\optc$-competitive order-oblivious algorithm for {\MSP} on co-graphic matroids.
\end{theorem}
\begin{proof}
Let us denote the input co-graphic matroid by $M$ and the maximum weight independent set of this matroid by $\OPT$. In this proof we consider an algorithm that executes the procedure of~\cite{soto_2013_matroid} on $M$, to get a random simple partition matroid $M'$ such that every element of $\OPT$ belongs to the optimal solution for $M'$ with probability $\nicefrac{1}{3}$. Then, it executes Algorithm~\ref{alg:partition} on $M'$. We observe that, by Theorem~\ref{thm:partition_algorithm}, every element that belongs to the optimal solution of $M'$ is selected by Algorithm~\ref{alg:partition} with probability at least $\nicefrac{1}{4}$. Combining this with the properties of \cite{soto_2013_matroid}'s procedure, we get that every element $e \in \OPT$ ends up in the output of the above algorithm with probability at least $\nicefrac{1}{12}$.
\end{proof}

It is worth to mention that the algorithm whose existence is guaranteed by the above theorem, like the original algorithm of~\cite{soto_2013_matroid} for co-graphic matroids, requires full access to the co-graphic matroid from the beginning. In other words, the algorithm knows the graph $G$ from the beginning, but learns the weight of every edge only when it arrives. Additionally, the above theorem implies, by the discussion in Section~\ref{sec:genToMatInt}, the following corollary.

\begin{corollary}
There exists a $(\nicefrac{1}{12}, 0)$-weakly-$\optc$-competitive order-oblivious algorithm for {\MSP} on co-graphic matroids, and this implies a $(1, \nicefrac{1}{12}, 0)$-reduce-and-solve algorithm for this problem.
\end{corollary}

\subsection{Transversal matroids}

A transversal matroid is defined by a bipartite graph $G = (N \cup V, E)$ in which the ground set $N$ forms one side. A set of nodes $I \subseteq N$ is independent in the transversal matroid if and only if there exists a matching in $G$ whose edges hit all nodes of $I$. A few algorithms with constant competitive ratios have been suggested for {\MSP} on transversal matroids. The first such algorithm was described by~\cite{dimitrov_2012_competitive}, and was later simplified and improved by~\cite{korula_2009_algorithms}. Both the original algorithm of~\cite{dimitrov_2012_competitive} and the improved algorithm of~\cite{korula_2009_algorithms} are order-oblivious, but are not known to be $\optc$-competitive (or even weakly-$\optc$-competitive) for any constant. In a later work, Kesselheim et al.~\cite{kesselheim_2013_optimal} came up with a very different algorithm for {\MSP} on transversal matroids. The algorithm of~\cite{kesselheim_2013_optimal} achieves an optimal competitive ratio of $\nicefrac{1}{e}$, but is, unfortunately, not order-oblivious.

Given the above survey of existing results, one can observe that currently there is no algorithm for {\MSP} on transversal matroids known to be order-oblivious and $c$-$\optc$-competitive for any constant $c$. Thus, we have no choice but to directly define here a reduce-and-solve algorithm $\cA$ for transversal matorids. We assume that $\cA$, like all the above mentioned algorithms for transversal matroids, knows the nodes of $V$ and also gets to know all the edges hitting a node $u \in N$ when this node arrives. The components of $\cA$ are given as follows:

\begin{itemize}
\item The ground set $N^\cA$ is the set $E$ of edges of the graph $G$, and the function $g^\cA \colon E(=N^\cA) \to N$ assigns every edge $e \in E$ to the single node of $N$ it hits. Note that for every $u \in N$ it is possible to evaluate $(g^\cA)^{-1}(u)$ after the arrival of $u$.\footnote{Technically, for $\cA$ to be a valid reduce-and-solve algorithm we need to guarantee that $(g^\cA)^{-1}(u)$ is not empty for every $u \in N$. This does not happen when there are nodes hitted by no edges, and thus, we need to create a dummy element $e_u$ for every such node $u$, add it to $N^\cA$, set $(g^\cA)^{-1}(u) = \{e_u\}$ and make $e_u$ a loop element of $M^\cA$. For simplicity, we ignore this technical issue.}
\item The matroid $M^\cA$ is a simple partition matroid over the ground set $E$ containing a single partition for every node of $V$. The partition of a node $v \in V$, which we denote by $\delta(v)$, contains the edges hitting it.
\item We need an order-oblivious algorithm $\bar{\cA}$ for {\MSP} on restrictions of $M^\cA$ containing exactly one edge from $(g^\cA)^{-1}(u)$ for every node $u \in N$. Moreover, we need $\bar{\cA}$ to have the extra property that, if we denote by $I$ its output set, then $g^\cA(I)$ is always independent in the original transversal matroid. Fortunately, we observe that in this case this extra property is meaningless because any independent set $I$ of such a restriction must be a matching of $G$ (no two edges in it hit a node $u \in N$ by the definition of the restriction, and no two edges in it hit a node $v \in V$ since it is independent in $M^\cA$), which implies that $g^\cA(I)$ is independent in the original transversal matroid. Hence, we can choose as $\bar{\cA}$ any order-oblivious algorithm for {\MSP} on simple partition matroids, and we choose it to be Algorithm~\ref{alg:partition}.
\end{itemize}

It remains to determine the parameters of the reduce-and-solve algorithm $\cA$. The fact that Algorithm~\ref{alg:partition} is $\nicefrac{1}{4}$-$\optc$-competitive by Theorem~\ref{thm:partition_algorithm} implies that $c^o = 1/4$ and $c^a = 0$. The next lemma shows that the last parameter ($c^r$) takes the value $1$ for this algorithm.

\begin{lemma}
For every independent set $B \subseteq N$ of the transversal matroid, there exists an independent set $R^\cA(B)$ of $M^\cA$ such that $R^\cA(B) \cap (g^\cA)^{-1}(u) \neq \varnothing$ for every node $u \in B$.
\end{lemma}
\begin{proof}
Since $B$ is independent, there must be a matching in $G$ whose edges hit all nodes of $B$. Let $R^\cA(B)$ be this matching. Note that, for every node $u \in B$, the set $(g^\cA)^{-1}(u)$ contains all edges hitting $u$, and thus, it must have a non-empty intersection with the matching $R^\cA(B)$ which contains such an edge by definition.
\end{proof}

\begin{corollary}
$\cA$ is a $(1, \nicefrac{1}{4}, 0)$-reduce-and-solve algorithm for {\MSP} on transversal matroids.
\end{corollary}

\subsection{Laminar matroids}

Consider a collection $\cL$ of subsets of a ground set $N$ and an upper bound $\mu(L)$ for every subset $L \in \cL$. One can define that a subset $I \subseteq N$ is independent if and only if $|I \cap L| \leq \mu(L)$ for every $L \in \cL$. It turns out that this way to define independence results in a matroid whenever the subsets of $\cL$ are laminar, i.e., any two subsets in $\cL$ are either disjoint or one of them contains the other. The matroid obtained this way is called a laminar matroid.

Ma et al.~\cite{ma_2013_simulatedSTACS} describe a $\nicefrac{1}{9.6}$-competitive algorithm for {\MSP} on laminar matroids which can be implemented using only independence oracle queries on sets of elements that already arrived. Moreover, the algorithm of~\cite{ma_2013_simulatedSTACS} is order-oblivious, and its analysis by~\cite{ma_2013_simulatedSTACS} shows that there exists a random set $R$ that depends only on the sample of the algorithm such that: $R$ is a subset of the output of the algorithm for every adversarial order in the second phase of the algorithm, and $R$ contains every element of $\OPT$ with probability at least $\nicefrac{1}{9.6}$ (where $\OPT$ is the maximum weight independent set in the laminar matroid). Combining these facts with our terminology, we get the following theorem.
\begin{theorem}
The algorithm of~\cite{ma_2013_simulatedSTACS} for {\MSP} on laminar matroids is a $\nicefrac{1}{9.6}$-$\optc$-competitive order-oblivious algorithm (and thus, also $(\nicefrac{1}{9.6}, 0)$-weakly-$\optc$-competitive). Moreover, it implies a $(1, \nicefrac{1}{9.6}, 0)$-reduce-and-solve algorithm for this problem.
\end{theorem}

\subsection{Regular and max-flow min-cut matroids}

In this section we show how our framework captures regular and max-flow min-cut matroids.

A regular matroid $M=(N,\mathcal{I})$ is a matroid that can be represented as a linear matroid (over the field $\mathbb{R}$) using a totally unimodular matrix (see Subsection~\ref{ssc:sparse} for a definition of linear matroids). We recall that a matrix is totally unimodular (TU) if and only if the determinant of every square submatrix of $A$ is either $-1$, $0$, or $1$.
An equivalent definition of a regular matroid is that a matroid is regular if and only if it can be represented as a linear matroid over any field (in fact, the representation by a TU matrix works simultaneously over any field). We refer the interested reader to~\cite{oxley_1992_matroid} for more information on regular matroids.

Dinitz and Kortsarz~\cite{dinitz_2013_matroid} presented a $\sfrac{1}{9e}$-competitive algorithm for {\MSP} on regular matroids, assuming that the matroid is known upfront, which we assume throughout this subsection.
In a preprocessing step, before observing or selecting any elements, their algorithm constructs a random matroid $M'=(N',\mathcal{I}')$ with the following properties:
\begin{enumerate}[label=(\roman*)]
\item $N' \subseteq N$ and $\mathcal{I}'\subseteq \mathcal{I}$, i.e., each independent set of $M'$ is independent in $M$.

\item\label{item:regRedSmallLoss} For each independent set $S\in \mathcal{I}$, there is a random set $R'(S)\subseteq \mathcal{I}'$, where the randomness is also over the random construction of $M'$, such that
\begin{equation*}
\Pr[e\in R'(S)] \geq \frac{3}{10} \qquad \forall e\in S\enspace.
\end{equation*}

\item The matroid $M'$ is a disjoint union of graphic and co-graphic matroids. Hence, there is a partition of the ground set $N'=N'_1 \cup N'_2 \cup \cdots \cup N'_q$, where $q\in \mathbb{Z}_{>0}$, and matroids $M'_i = (N'_i,\mathcal{I}'_i)$ for $i\in [q]$ such that
\begin{itemize}[nosep]
\item $M'_i$ is graphic or co-graphic for each $i\in [q]$,
\item $\mathcal{I}' = \{I_1\cup \cdots \cup I_q \mid I_i \in \mathcal{I}'_i \text{ for } i\in [q]\}$.
\end{itemize}

\end{enumerate}

In words, the matroid $M'$ can be interpreted as a more restrictive version of $M$, leading to a much simpler structure, but losing a factor of at most $\sfrac{10}{3}$ in terms of the maximum weight independent set. 
The {\MSP} algorithm of~\cite{dinitz_2013_matroid} first constructs the matroid $M'$, and then simply applies in parallel a $\sfrac{1}{2e}$-competitive algorithm for all graphic matroids in the decomposition and a $\sfrac{1}{3e}$-competitive algorithm for all co-graphic matroids. A direct analysis would lead to a competitive ratio of $\sfrac{1}{10e}$, which follows from the factor $\sfrac{3}{10}$ loss incurred due to~\ref{item:regRedSmallLoss} multiplied by the competitive ratio of $\sfrac{1}{3e}$ for co-graphic matroids.
The competitive ratio of $\sfrac{1}{9e}$ obtained in~\cite{dinitz_2013_matroid} follows by a more careful analysis. We briefly address this in a comment at the end of this subsection.


To obtain a reduce-and-solve algorithm for regular matroids, we can simply first construct the matroid $M'$ with the above properties as described in~\cite{dinitz_2013_matroid}, and then run in parallel our previously-described reduce-and-solve algorithms for graphic and co-graphic matroids.
A technicality we have to watch out for, is to make sure that we stay in the order-oblivious model. More precisely, when running order-oblivious secretary algorithms in parallel, then each algorithm may want to query a certain number of elements of its ground set during the sampling phase. However, to make sure that the algorithm which corresponds to the parallel run of several algorithms is still order-oblivious, we can only decide about $m$ uniformly random elements that will be sampled overall. There are different ways around this problem. One simple way to resolve this is to make sure that both algorithms we use for graphic and for co-graphic matroids observe, during their sampling phase, each element of their ground set with probability $0.5$ (or any other common value within $(0,1)$). This way, we can simply observe each element of the combined ground set with probability $0.5$. Our $(1,\sfrac{1}{12},0)$-reduce-and-solve algorithm for co-graphic matroids already has this property. The algorithm we suggested for graphic matroids (see Corollary~\ref{cor:graphicMat}) observes each element with probability $p=2/3$. However, we can run the same algorithm for $p=1/2$, which, by Lemma~\ref{lem:graphicMatExt}, leads to a slightly weaker $(1,\sfrac{1}{8},0)$-reduce-and-solve algorithm for graphic matroids that observes each element with probability $0.5$ in the sampling phase. Hence, by running these algorithms in parallel, we obtain a $(1,\sfrac{1}{12},0)$-reduce-and-solve algorithm for any matroid that is a disjoint union of graphic and co-graphic matroids. 
Moreover, we loose a factor of $\sfrac{10}{3}$ due to the reduction to the matroid $M'$. It is not hard to observe that the reduction to $M'$ done at the beginning of the algorithm can be interpreted as part of the reduce-and-solve algorithm, leading to a $(\sfrac{3}{10},\sfrac{1}{12},0)$-reduce-and-solve procedure.

\begin{theorem}
There is a $(\sfrac{3}{10},\sfrac{1}{12},0)$-reduce-and-solve algorithm for {\MSP} on regular matroids.
\end{theorem}

Furthermore, it was shown in~\cite{dinitz_2013_matroid} that their technique naturally extends to a generalization of regular matroids, known as max-flow min-cut matroids (see~\cite{seymour_1977_matroids} for a formal definition). The approach is analogous to the one used for regular matroids, with the difference that instead of losing a factor of $\sfrac{3}{10}$ through the reduction to $M'$, a factor of $\sfrac{2}{7}$ is lost. By the same reasoning as above we therefore get.

\begin{theorem}
There is a $(\sfrac{2}{7},\sfrac{1}{12},0)$-reduce-and-solve algorithm for {\MSP} on max-flow min-cut matroids.
\end{theorem}

\noindent \textbf{Remark:} 
The algorithms presented in~\cite{dinitz_2013_matroid} do the reduction from $M$ to $M'$ in two steps. In a first step, a matroid $\bar{M}$ is obtained that is a disjoint union of graphic, co-graphic, and one additional simple type of matroid, which depends on whether $M$ is a regular or a max-flow min-cut matroid. This reduction only loses a factor of $\sfrac{1}{3}$. In a second step, the additional type of matroid is reduced to a graphic matroid at a loss of either $\sfrac{9}{10}$ for regular matroids or $\sfrac{6}{7}$ for max-flow min-cut matroids. However, this additional loss is compensated by the fact that the competitive ratio of graphic matroids is better than for co-graphic matroids, leading to an overall competitiveness of $\sfrac{1}{9e}$ for regular matroids and max-flow min-cut matroids. With a similar analysis we could also avoid losing a factor of $\sfrac{9}{10}$ for regular matroids and $\sfrac{6}{7}$ for max-flow min-cut matroids, respectively. For clarity we decided to present here a slightly weaker, but simpler, reduction.

\subsection{General matroids}

In this section we describe a weakly-$\optc$-competitive algorithm for {\MSP} on a general matroid constraint $M = (N, \cI)$ which is based on the algorithm of~\cite{feldman_2015_simple} for {\MSP}. We begin by recapping some of the details of the last algorithm.

One can interpreted the algorithm of~\cite{feldman_2015_simple} as an algorithm getting two parameters $w_{\max}$ and $h$, and then defining $h$ weight classes $C_1, \ldots C_h$ by
\begin{equation*}
C_i = \left\{
e\in N \;\middle\vert\; w(e) \in \left.\left(\frac{w_{\max}}{2^{h-i+1}}, \frac{w_{\max}}{2^{h-i}}\right.\right]
\right\} \qquad \forall\; i = 1, \dotsc, h \enspace.
\end{equation*}
After defining these weight classes, the algorithm executes one of two procedures, which, for simplicity, we simply call Procedure~A and Procedure~B. More precisely, the algorithm decides at the beginning at random, with well-defined probabilities, which of the two procedures to run. For simplicity, we can assume that each one of the two procedures is chosen with probability $\nicefrac{1}{2}$. This is not the probability employed in~\cite{feldman_2015_simple}, but the change in the probability only changes the guarantee of the algorithm by a constant factor. For the interested readers, we note that the algorithm of~\cite{feldman_2015_simple} is parameterized by a parameter $\tau\in \{0,\dots, \lceil\log(h+1)\rceil \}$ which is chosen uniformly at random. What we denote as Procedure~A is the algorithm obtained by choosing $\tau$ uniformly at random within $\{1,\ldots, \lceil \log(h+1) \rceil\}$, whereas Procedure~B corresponds to letting the algorithm of~\cite{feldman_2015_simple} run with $\tau = 0$.

Both Procedure~A and Procedure~B first observe a subset $S$ containing each element of $N$ with probability $\nicefrac{1}{2}$ (equivalently, $S$ is a sample $m$ elements, where $m$ is distributed according to the binomial distribution $B(|N|, $\nicefrac{1}{2}$)$). Then, they select elements of $N \setminus S$ in a way that we describe later. The following is a key quantity in the analysis of both procedures.
\begin{equation*}
p_{e,i} \coloneqq \Pr[e\in \spn(S \cap C_{\geq i}) \mid e\not\in S]
\qquad \forall e\in N, i = 1,\dotsc, h \enspace,
\end{equation*}
where $C_{\geq i}$ is the union of all the classes $C_i, C_{i+1}, \dotsc, C_h$.

\paragraph{Selection by Procedure~A.}

Based on $S$, Procedure~A randomly constructs a matroid $M^A=(N^A,\mathcal{I}^A)$, with $N^A \subseteq N\setminus S$ and $\mathcal{I}^A\subseteq \mathcal{I}$, and then greedily accepts any appearing element that, together with the already accepted elements, forms an independent set in $M^A$. The key power of Procedure~A is that $M^A$ has the following property for every $i = 1,\dotsc, h$ and $e\in C_i$.
\[
	\Pr\mathopen{}\left[e \text{ is a free element in $M^A$}\right]
	\geq
	\frac{1-p_{e,i}}{8 \lceil\log(h+1)\rceil}
	\enspace,
\]
where the probability is with respect to the randomness of $S$ and the random construction of $M^A$. We recall that an element is \emph{free} in a matroid if it is not spanned by all the other elements together. Hence, a free element can be added to any independent set without destroying independence.
Consequently, a free element in $M^A$ is always selected by Procedure~A, regardless of the order in which the elements of $N\setminus S$ appear.

\paragraph{Selection by Procedure~B}

Like Procedure~A, Procedure~B also randomly constructs a matroid $M^B=(N^B, \mathcal{I}^B)$ based on $S$ with $N^B\subseteq N\setminus S$ and $\mathcal{I}^B\subseteq \mathcal{I}$, and then greedily selects appearing elements, maintaining an independent set in $M^B$. However, the properties of $M^B$ are different. The first property that $M^B$ has is that it is a disjoint union of $h$ matroids, where each one of the matroids is defined over $N^B\cap C_i$ for a different $i = 1, \dotsc, h$. Consequently, the selection of an element of $N^B\cap C_i$ for some $i$ does not affect the set of elements that can be select from $N^B \cap C_{i'}$ for any $i' \neq i$. The second property of $M^B$ is that it has a random independent set $I^B$ such that

\begin{equation*}
	\Pr[e\in I^B] = \frac{p_{e,i}}{4} 
\qquad \forall\; i = 1,\dotsc, k \text{ and } e\in C_i\cap \OPT \enspace,
\end{equation*}
where the probability is, again, with respect to the randomness of $S$ and the random construction of $M^B$.

We now would like to run Procedures~A and~B against a switching adversary. The switching adversary places each element selected by the procedures into one of two sets $T_1$ and $T_2$, and we assume that Procedure~X (where $X$ is either $A$ or $B$) adds an arriving element to the solution if and only if the addition of this element to $T_1$ will not destroy the independence in $M^X$ of $T_1$. Observe that $T_1$ always remains independent in $M$ since $\cI^X \subseteq \cI$. The following lemma shows that we already have an algorithm which is close to be weakly-$\optc$-competitive.

\begin{lemma} \label{lem:old_algorithm}
Executing Procedures~A or~B, with probability $\nicefrac{1}{2}$ each, against a switching adversary satisfies the following inequality.
\[
	w((T_1 \cup T_2) \cap O) + 2 \cdot w(T_1 \setminus O) \geq w(O)
	\enspace,
\]
where $O$ is a set, depending only on $S$ and the random coins used by the procedures, containing each element of $\OPT \cap (C_1, \dotsc, C_h)$ with probability at least $1/O(\log h)$.
\end{lemma}
\begin{proof}
First, we define $O$ to be the set of elements of $\OPT \cap (C_1, \dotsc, C_h)$ which are either free in $M^A$ if Procedure~A is applied or appear in $I^B$ if Procedure~B is applied. According to the properties of Procedures~A and~B described above, the probability of every element of $\OPT \cap (C_1, \dotsc, C_h)$ to belong to $O$ is at least
\[
	\frac{1}{2} \cdot \frac{1-p_{e,i}}{8 \lceil\log(h+1)\rceil} + \frac{1}{2} \cdot \frac{p_{e,i}}{4}
	=
	\frac{1}{O(\log h)}
	\enspace.
\]

There are now two cases to consider. If Procedure~A is executed, then all the elements of $O$ are selected because they are free elements of $M^A$, which implies $O \subseteq T_1 \cup T_2$. Hence, we get
\[
	w((T_1 \cup T_2) \cap O) + 2 \cdot w(T_1 \setminus O)
	\geq
	w((T_1 \cup T_2) \cap O)
	=
	w(O)
	\enspace.
\]
Consider now the case that Procedure~B is executed, and recall that $O \subseteq I_B$ is independent in $M^B$ and $M^B$ has the property that the selection of elements from one class does affect the possibility to select elements from the other classes. These facts imply that if we denote the set of elements of $O \cap C_i$ which are not selected by $O_i$, then there must be at least $|O_i|$ elements from $C_i \setminus O$ which appear in $T_1$ and prevented the selection of the elements of $O_i$. Since the weights of all the elements of $C_i$ are equal up to a factor of $2$, this implies
\[
	2 \cdot w((T_1 \cap C_i) \setminus O)
	\geq
	w(O_i)
	\quad
	\forall \; i = 1, \dotsc, k
	\enspace.
\]
Using this inequality we can now get
\begin{align*}
	w((T_1 \cup T_2) \cap O) + 2 \cdot w(T_1 \setminus O)
	={} &
	w((T_1 \cup T_2) \cap O) + 2 \cdot \sum_{i = 1}^h w((T_1 \cap C_i) \setminus O)\\
	\geq{} &
	w\left(O \setminus \bigcup_{i = 1}^h O_i\right) + \sum_{i = 1}^h w(O_i)
	=
	w(O)
	\enspace.
	\qedhere
\end{align*}
\end{proof}

There are still a few issues that we need to take care of in order to get a true weakly-$\optc$-competitive algorithm. Specifically, we need to find a way to set values for $w_{\max}$ and $h$, and we need to guarantee that all the elements of $\OPT$ outside the weight classes $C_1, \dots, C_h$ belong to $L(\OPT)$. Algorithm~\ref{alg:weakly_general} takes care of these issues.

\SetKwIF{With}{OtherwiseWith}{Otherwise}{with}{do:}{otherwise with}{otherwise}{}
\begin{algorithm2e}[ht]
\caption{\textsf{Weakly-$\optc$-competitive Algorithm for General Matroids}} \label{alg:weakly_general}
Pick $m$ according to the binomial distribution $B(|N|, \nicefrac{1}{2})$.\\
Request a sample $S$ of $m$ elements.\\
Let $W$ be the maximum weight of a non-loop element in $S$, and let $\tilde{\rho} \leftarrow 10 \cdot \rank(S)$.\\
\With{probability $\nicefrac{1}{2}$}
{
	Pick the first non-loop element of $N \setminus S$ whose value is at least $W$. \label{line:noprocedures}\\
}
\Otherwise
{
	Set $h = \lceil\log_2(\tilde{\rho} + 1)\rceil$ and $w_{\max} = W$, and execute Procedures~A or~B, with probability $\nicefrac{1}{2}$ each, on the matroid $M|_{N \setminus S}$. \label{line:procedures}
}
\end{algorithm2e}

\noindent \textbf{Remark:} Technically, as given, Algorithm~\ref{alg:weakly_general} is not order-oblivious since it requests the sample set $S$, and then Procedures~A and~B might require a second sample from $N \setminus S$. However, since the event that these procedures require a second sample is independent of the content of $S$, this technical issue can be fixed by taking a larger sample (containing every element of $N$ with probability $\nicefrac{3}{4}$) when a second sample is necessary, and then splitting it randomly into the two samples. We omit the very technical details of the split.

Observe that the definition of a weakly-$\optc$-competitive algorithm requires nothing in the case that $\OPT = \varnothing$. Thus, we may assume in the rest of this section, for simplicity, that $\OPT \neq \varnothing$. Let us now introduce some additional notation. Recall that $h(\OPT)$ is the heaviest element in $\OPT$, and let $s(N)$ be the second heaviest non-loop element in $N$ (the first must be $h(\OPT)$ itself). Next, let $\cE_1$ be the event that $h(\OPT) \not \in S$, $s(N) \in S$ and Algorithm~\ref{alg:weakly_general} gets to Line~\ref{line:noprocedures} (if $s(N)$ does not exist, i.e., $N$ contains only one non-loop element, then we drop the requirement $s(N) \in S$ from the definition of $\cE_1$). Additionally, let $\cE_2$ be the event that $h(\OPT) \in S$, $\tilde{\rho} \geq |\OPT|$ and Algorithm~\ref{alg:weakly_general} gets to Line~\ref{line:procedures}. It is easy to see that $\Pr[\cE_1] \geq 1/8$. Bounding the probability of $\cE_2$ is more involved, and it is done by the next lemma.

\begin{lemma} \label{lem:e_2_probability}
Fix any element $e \in \OPT - h(\OPT)$. Then, $\Pr[\cE_2 \mid e \not \in S] \geq 0.2$.
\end{lemma}
\begin{proof}
The event $\cE_2$ requires both that $h(\OPT)$ belongs to $S$ and that Algorithm~\ref{alg:weakly_general} gets to Line~\ref{line:procedures}. Each one of these conditions holds with probability $\nicefrac{1}{2}$, independently, and the second of them is also independent of $\tilde{\rho}$. Thus,
{\allowdisplaybreaks
\begin{align*}
	\Pr[\cE_2 \mid e \not \in S]
	={} &
	\frac{1}{4} \Pr[\tilde{\rho} \geq |\OPT| \mid e \not \in S, h(\OPT) \in S]\\
	={} &
	\frac{1}{4} \Pr[10 \cdot \rank(S) \geq |\OPT| \mid e \not \in S, h(\OPT) \in S]\\
	\geq{} &
	\frac{1}{4} \Pr[10 \cdot  \rank(S \cap \OPT) \geq |\OPT| \mid e \not \in S, h(\OPT) \in S]\\
	={} &
	\frac{1}{4} \Pr[10 \cdot  |S \cap \OPT| \geq |\OPT| \mid e \not \in S, h(\OPT) \in S]\\
	={} &
	\frac{1}{4} \Pr[10 \cdot |S \cap (\OPT \setminus \{e, h(\OPT)\})| \geq |\OPT| - 10]
	\enspace,
\end{align*}
}%
where the penultimate equality holds since $S \cap \OPT$ is an independent set, and the last equality holds since each element belongs to $S$ with probability $\nicefrac{1}{2}$, independently. It remains to prove that
\[
	\Pr[10 \cdot |S \cap (\OPT \setminus \{e, h(\OPT)\})| \geq |\OPT| - 10] \geq 0.8
	\enspace,
\]
which is equivalent to
\begin{equation} \label{eq:objective}
	\Pr[10 \cdot |S \cap (\OPT \setminus \{e, h(\OPT)\})| < |\OPT| - 10] \leq 0.2
	\enspace.
\end{equation}
Observe that
\[
	\E[10 \cdot |S \cap (\OPT \setminus \{e, h(\OPT)\})|]
	=
	10 \cdot \frac{|\OPT \setminus \{e, h(\OPT)\}|}{2}
	=
	5 \cdot (|\OPT| - 2)
	\enspace.
\]
Thus, Inequality~\eqref{eq:objective} can be rewritten as
\[
	\Pr\mathopen{}\left[10 \cdot |S \cap (\OPT \setminus \{e, h(\OPT)\})| < \left(1 - \frac{4}{5}\right) \cdot \E[10 \cdot |S \cap (\OPT \setminus \{e, h(\OPT)\})|] - 8\right] \leq 0.2
	\enspace.
\]
To see why this inequality is true, we observe that $|S \cap (\OPT \setminus \{e, h(\OPT)\})|$ can be viewed as the sum of $|\OPT| - 2$ independent random variables taking values from the set $\{0, 1\}$, which implies by Lemma~2.3 of~\cite{badanidiyuru_2014_fast} that the left hand side of the last inequality is upper bounded by
\[
	e^{-\frac{(4/5) \cdot 8}{2}}
	=
	e^{-3.2}
	<
	0.2
	\enspace.
	\qedhere
\]
\end{proof}

Let $\hat{O}$ be a set which takes a value depending on the events $\cE_1$ and $\cE_2$. When $\cE_1$ happens $\hat{O} = \{h(\OPT)\}$. When $\cE_2$ happens $\hat{O} = O \cap \OPT$, where $O$ is the set whose existence is guaranteed by Lemma~\ref{lem:old_algorithm} for the execution of Procedures~A or~B on the matroid $M|_{N \setminus S}$ (recall that $\cE_2$ implies that Algorithm~\ref{alg:weakly_general} gets to Line~\ref{line:procedures}, and thus, executes one of the Procedures~A or~B). Finally, when neither $\cE_1$ nor $\cE_2$ happen, we define $\hat{O}$ to be the empty set. We observe that $\hat{O}$ is a function of the random bits used by Algorithm~\ref{alg:weakly_general} and Procedures~A and~B and on the samples used by these algorithms.

\begin{lemma} \label{lem:necessary_equality_general}
When Algorithm~\ref{alg:weakly_general} is executed against a switching adversary the following inequality is satisfied:
\[
	w((T_1 \cup T_2) \cap \hat{O}) + 2 \cdot w(T \setminus \hat{O}) \geq w(\hat{O})
	\enspace.
\]
\end{lemma}
\begin{proof}
The lemma is trivial when $\hat{O} = \varnothing$. Thus, we only need to consider the cases that either $\cE_1$ or $\cE_2$ happen. Let us consider first the case that $\cE_1$ occurs. By definition, $\cE_1$ implies $h(\OPT) \not \in S$. Additionally, one can observe that it also implies that $W$ is set to $s(N)$ (when $s(N)$ exists), and thus, it guarantees that Algorithm~\ref{alg:weakly_general} picks $h(\OPT)$ when it gets to Line~\ref{line:noprocedures}. Hence, $\cE_1$ implies $T_1 \cup T_2 = \{h(\OPT)\}$, and we get
\[
	w((T_1 \cup T_2) \cap \hat{O}) + 2 \cdot w(T_1 \setminus \hat{O})
	\geq
	w((T_1 \cup T_2) \cap \hat{O})
	=
	w(\{h(\OPT)\})
	=
	w(\hat{O})
	\enspace.
\]
Consider now the case that $\cE_2$ occurs. In this case we have, by Lemma~\ref{lem:old_algorithm},
\[
	w((T_1 \cup T_2) \cap O) + 2 \cdot w(T_1 \setminus O) \geq w(O)
	\Rightarrow
	2 \cdot w(T_1 \setminus O) \geq w(O \setminus (T_1 \cup T_2))
	\enspace.
\]
Since $\hat{O} = O \cap \OPT \subseteq O$, the last inequality implies
\[
	2 \cdot w(T_1 \setminus \hat{O}) \geq w(\hat{O} \setminus (T_1 \cup T_2))
	\Rightarrow
	w((T_1 \cup T_2) \cap \hat{O}) + 2 \cdot w(T_1 \setminus \hat{O}) \geq w(\hat{O})
	\enspace.
	\qedhere
\]
\end{proof}

\begin{lemma} \label{lem:necessary_probabilities_general}
For every element $e \in \OPT \setminus L(\OPT)$,
\[
	\Pr[e \in \hat{O}]
	\geq
	\frac{1}{O(\log \log \rank(M))}
	\enspace.
\]
\end{lemma}
\begin{proof}
There are two cases to consider. First, consider the case that $e = h(\OPT)$. Observe that, by the definition of $\hat{O}$,
\[
	\Pr[e \in \hat{O}]
	\geq
	\Pr[\cE_1]
	\geq
	\frac{1}{8}
	\enspace.
\]

It remains to consider the case that $e \in (\OPT - h(\OPT)) \setminus L(\OPT)$. Note that this case implies $|\OPT| \geq 2$, and let us prove that the inequality $w(e) \geq w(h(\OPT)) / |\OPT|$ must hold in this case. Assume towards a contradiction that this inequality does not hold, and recall that $L(\OPT)$ contains all the elements of $\OPT$ whose weight is at most some threshold $\ell$, where the threshold is selected to be the maximal threshold in $\{w(e) \mid e \in N\}$ such that the total weight of $L(\OPT)$ is less than $h(\OPT)$. By our assumption and the fact that $e \neq h(\OPT)$, we get
\[
	w(\{e' \in \OPT \mid w(e') \leq w(e)\}
	<
	|\{e' \in \OPT - h(\OPT)\}| \cdot \frac{w(h(\OPT))}{|\OPT|}\\
	<
	w(h(\OPT))
	\enspace.
\]
Hence, setting the threshold $\ell$ to $w(e)$ would have kept the total weight of $L(\OPT)$ no larger than $w(h(\OPT))$, and thus, $\ell$ must be at least as large as $w(e)$. However, this leads to a contradiction since the fact that $e \not \in L(\OPT)$ implies $\ell < w(e)$.

Next, let us denote by $\OPT'$ the optimal solution of $M_{N \setminus S}$, and consider now what happens when the event $\cE_2$ happens and $e \not \in S$. These events imply together a few things. First, $\cE_2$ implies that $w_{\max} = W = w(h(\OPT))$ since $h(\OPT) \in S$. Second, we get that $e \in \OPT'$ since an element which belongs to the maximum weight independent set of a matroid must also belong to the corresponding set in every restriction of the matroid containing it. Finally, we also get $h = \lceil \log(\tilde{\rho} + 1)\rceil \geq \lceil \log_2(|\OPT| + 1)\rceil$, which implies that $e$ belongs to one of the weight classes $C_1, C_2, \dotsc, C_h$ because any element belongs to these classes if its weight is at most $w_{\max} = w(h(\OPT))$ and larger than
\[
	\frac{w_{\max}}{2^h}
	\leq
	\frac{w(h(\OPT))}{2^{\log_2(|\OPT| + 1)}}
	=
	\frac{w(h(\OPT))}{|\OPT| + 1}
	<
	\frac{w(h(\OPT))}{|\OPT|}
	\enspace.
\]

By Lemma~\ref{lem:old_algorithm}, the above observations imply together that, conditioned on $\cE_2$ and $e \not \in S$, $O$ contains $e$ with probability $\Omega(1 / \log h) = \Omega(1 / \log \log \rank(M))$, where the equality holds since
\[
	h
	=
	\lceil \log_2(\tilde{\rho} + 1) \rceil
	=
	\lceil \log_2(10 \cdot \rank(S) + 1) \rceil
	\leq
	\lceil \log_2(10 \cdot \rank(M) + 1) \rceil
	=
	O(\log \rank(M))
	\enspace.
\]
Since $O$ is a subset of $\hat{O}$ when $\cE_2$ happens, this yields
\[
	\Pr[e \in \hat{O}]
	\geq
	\frac{\Pr[\cE_2 \text{ and } e\not \in S]}{O(\log \log \rank(M))}
	=
	\frac{\Pr[\cE_2 \mid e \not \in S] \cdot \Pr[e \not \in S]}{O(\log \log \rank(M))}
	\geq
	\frac{0.1}{O(\log \log \rank(M))}
	\enspace,
\]
where the last inequality follows from Lemma~\ref{lem:e_2_probability}.
\end{proof}

Combining Lemmata~\ref{lem:necessary_equality_general} and~\ref{lem:necessary_probabilities_general}, we get the following corollary.

\begin{corollary}
Algorithm~\ref{alg:weakly_general} is a $(1/O(\log \log \rank(M)), 2)$-weakly-$\optc$-competitive algorithm for {\MSP}. Hence, there exists a $(1, 1 / O(\log \log \rank(M)), 2)$-reduce-and-solve algorithm for this problem.
\end{corollary}

\section{Submodular MISP}\label{sec:submodular}

In this section we prove our reduction from the submodular version of {\MISP} to its linear version, i.e., Theorem~\ref{thm:mainSubm}. This proof is an extension of an approach used in~\cite{feldman_2015_submodular}, which implied a reduction from submodular MSP to linear MSP.
Throughout this section, let $M_i=(N,\mathcal{I}_i)$ for $i\in [k]$ be $k$ matroids on a common ground set $N$, and let $f$ be a nonnegative submodular function over $N$. We start by stating some properties of the greedy algorithm for (offline) submodular function maximization (SFM). We recall that the greedy algorithm for SFM over some constraint family $\mathcal{F}$, which, in our case, will be $\mathcal{F}=\bigcap_{i=1}^k \mathcal{I}_i$, is defined analogously as in the case of a linear function $f$, and thus works as follows:

\medskip

\quad\parbox[t]{0.8\linewidth}{
Start with the empty set $S=\varnothing$. As long as there is an element $e\in N\setminus S$ such that $f(S+e)-f(S) >0$, choose such an element $e$ maximizing $f(S+e)-f(S)$ and add it to $S$. As soon as no such element exists anymore, return $S$.
}

\bigskip

As in the case of linear objective functions, we denote by $\greedy(N)$, or when we want to be more specific by $\greedy_f(N)$, the set $S$ returned by the greedy algorithm. Moreover, for sets $A,B\subseteq N$ we use the notation $f_A(B) = f(A\cup B) - f(A)$ for the marginal value of $B$ with respect to $A$.

Moreover, following the approach in~\cite{feldman_2015_submodular}, we use \emph{convolutions} of $f$. More precisely, for a vector $w\in \mathbb{R}^N$, the function $f_w$ is defined by
\begin{equation*}
f_w(S) = \min_{A\subseteq S}\{f(A) + w(S\setminus A)\} \qquad \forall S\subseteq N\enspace.
\end{equation*}
This is a well-known construction in the field of submodular optimization, and leads to a function $f_w$ with many interesting properties. In particular, if $f$ is nonnegative submodular and $w\geq 0$, then $f_w$ is also a nonnegative submodular function. For more information, we refer the interested reader to the discussion in~\cite{feldman_2015_submodular} and references therein.

We note that, for simplicity, we do not present here the strongest possible analysis. In particular, we present the submodular to linear reduction independently of our framework for {\MISP}. One could merge this reduction with the procedure we use in the design of {\MISP} algorithms to make the optimal solutions overlap, and this way obtain better competitive ratios.

\subsection{The algorithm to reduce submodular MISP to linear MISP}

The algorithm we suggest is identical to the one used in~\cite{feldman_2015_submodular} with the only exception that we deal with matroid intersection constraints $\mathcal{F}=\cap_{i=1}^k \mathcal{I}_i$ instead of a single matroid constraint. Algorithm~\ref{alg:online} repeats the algorithm presented in~\cite{feldman_2015_submodular} adjusted to our context. We denote by $\mathcal{A}$ an $\alpha$-competitive algorithm for the linear MISP problem over $\mathcal{F}$. Algorithm~\ref{alg:online} describes how, after having observed each element with probability $0.5$, the remaining elements get passed to $\mathcal{A}$ in a black-box way. To best match our description to the one used in~\cite{feldman_2015_submodular}, Algorithm~\ref{alg:online} is described in the random arrival model, where elements arrive one-by-one in a uniformly random random. However, it is clear from its description that Algorithm~\ref{alg:online} is an order-oblivious reduction, in the sense that it can be used with any order-oblivious algorithm $\mathcal{A}$ in an order-oblivious model.
Notice that the set $E$ constructed in Algorithm~\ref{alg:online} is the set of greedy-relevant elements with respect to $X$, i.e., $E=\imp(X)$.

For simplicity, our description of Algorithm~\ref{alg:online} only passes elements to $\mathcal{A}$ which are not part of the sample $L$. Strictly speaking, this would require that $\mathcal{A}$ remains $\alpha$-competitive even if only a subset of all elements of the ground set are passed to it. However, this can easily be resolved by also passing the elements of $L$ to $\mathcal{A}$ in a uniformly random order, such that the appearance order of any element of $N$ is uniformly at random.
This reasoning is analogous to the one we used in Section~\ref{sec:makeOptsOverlap} as part of our MISP to MSP reduction.

\SetKwIF{With}{OtherwiseWith}{Otherwise}{with}{do}{otherwise with}{otherwise}{}
\begin{algorithm2e}[ht]
\caption{$\online(p)$} \label{alg:online}
\tcp{Learning Phase}
Choose $X$ from the binomial distribution $B(n, 1/2)$.\\
Observe (and reject) the first $X$ elements of the input. Let $L$ be the set of these elements.\\

\BlankLine

\tcp{Selection Phase}
Let $G$ be the output of $\greedy$ on the set $L$.\\
Let $E \gets \varnothing$.\\
\For{each arriving element $u \in N \setminus L$}
{
	Let $w(u) \gets 0$.\\
	\If{$u \in \imp(G)$}
	{
		\With{probability $p$}
		{
			Add $u$ to $E$.\\
			Let $G_u \subseteq G$ be the solution of $\greedy$ immediately before it adds $u$ to it.\\
			Update $w(u) \gets f(u \mid G_u)$.\\
		}
	}
	Pass $u$ to $\mathcal{A}$ with weight $w(u)$.
}
\Return{$Q \cap E$, where $Q$ is the output of $\mathcal{A}$}.
\end{algorithm2e}

In the rest of this section, we prove the theorem below, which shows that this reduction indeed allows for obtaining a guarantee on the competitive ratio for Submodular MISP from the $\alpha$-competitiveness of $\mathcal{A}$. Moreover, Theorem~\ref{thm:mainSubmToLin} implies our main theorem on Submodular MISP (Theorem~\ref{thm:mainSubm}).

\begin{theorem} \label{thm:mainSubmToLin}
By choosing $p=\sfrac{\alpha}{3k}$ in Algorithm~\ref{alg:online}, we obtain
\begin{equation*}
\E[f(Q\cap E)] \geq \frac{\alpha^2}{128 k^2}\cdot f(\OPT) \enspace.
\end{equation*}
\end{theorem}

\subsection{Analysis of Algorithm~\ref{alg:online}}

\medskip

The following lemma is a slight generalization of Lemma~3 in~\cite{gupta_2010_constrained}, the only difference being that $f$ is not assumed to be normalized (i.e., $(\varnothing)=0$).\footnote{Moreover the Lemma in~\cite{gupta_2010_constrained} is stated for $k$-independence systems, which is a generalization of the intersection of $k$ matroids. Since we only need the statement for the intersection of $k$ matroids, we state it here in this form for simplicity.}
The proof presented in a long version~\cite{gupta_2010b_constrained} of the paper~\cite{gupta_2010_constrained} easily carries over to non-normalized functions, and we therefore do not repeat it here.%


\begin{lemma}[See~\cite{gupta_2010b_constrained}]\label{lem:SFMGreedyProp}
The solution $S=\greedy_f(N)$ returned by the greedy algorithm for the maximization of a nonnegative submodular function over the intersection of $k$ matroids satisfies
\begin{equation*}
f(S) \geq \frac{1}{k+1} \cdot f(C\cup S) \qquad \text{for any $C\in \bigcap_{i=1}^k \mathcal{I}_i$}\enspace.
\end{equation*}
\end{lemma}

Lemma~\ref{lem:SFMGreedyProp} immediately implies that for monotone submodular functions $f$, the greedy algorithm is a $(k+1)$-approximation algorithm (a fact which was known since the late 1970's~\cite{fisher_1978_analysis}).
However, analogously to the approach in~\cite{feldman_2015_submodular}, we want to obtain results for nonmonotone submodular function maximization, and it is well known that the greedy algorithm is not a constant-factor approximation algorithm for such objectives even on a single matroid constraint.
The following is a generalization of Lemma IV.3 of~\cite{feldman_2015_submodular}, which was originally presented in the context of applying the greedy algorithm to a single matroid constraint, to the intersection of $k$ matroids.
It shows that letting the greedy algorithm run on a randomly sampled subset provides approximation guarantees even for nonmonotone SFM.
Again, the proof is identical with the only difference that we invoke Lemma~\ref{lem:SFMGreedyProp} for a general $k$ instead of $k=1$.

\begin{lemma}[see~\cite{feldman_2015_submodular}]\label{lem:fmValue}
Let $S$ be a random set containing every element of $N$ with probability $1/2$, independently, then
\begin{equation*}
\E[f(\greedy_f(S))] \geq \frac{1}{4(k+1)} \cdot f(\OPT)\enspace.
\end{equation*}
\end{lemma}

Notice that the set $G$ constructed in Algorithm~\ref{alg:online} is indeed obtained by applying $\greedy$ to a random subset of $N$ containing each element of $N$ with probability $0.5$, independently. Lemma~\ref{lem:fmValue} thus implies the following.
\begin{corollary}\label{cor:gLargeWrtOpt}
\begin{equation*}
\E[f(G)] \geq \frac{1}{4(k+1)}\cdot f(\OPT)\enspace.
\end{equation*}
\end{corollary}

The proof of the following two lemmata from~\cite{feldman_2015_submodular} do not depend on $\mathcal{F}$ being the independent sets of a matroid, and therefore, also hold in our setting.

\begin{lemma}[Lemma IV.5 of \cite{feldman_2015_submodular}]
\begin{equation*}
\E[w(E)] = p\cdot \E[w(G)]\enspace.
\end{equation*}
\end{lemma}

\begin{lemma}
\begin{equation*}
\E[f_w(E)] \geq \frac{f(\varnothing)}{1+p} + \frac{p(1-p)}{1+p}\cdot \E[f_w(G)]\enspace.
\end{equation*}
\end{lemma}

The following observation is implied by the fact that $\mathcal{A}$ is $\alpha$-competitive for linear objectives. The term $\OPT_w(E)$ denotes here the maximum weight independent subset of $E$ with respect to the weight function $w$.

\begin{observation}[Observation IV.7 of~\cite{feldman_2015_submodular}]\label{obs:linRatio}
\begin{equation*}
\E[w(Q\cap E)] = \E[Q] \geq \alpha\E[w(\OPT_w(E))]\enspace.
\end{equation*}
\end{observation}

The next lemma is a generalization of Lemma~IV.8 in~\cite{feldman_2015_submodular} from one to $k$ matroids. This lemma is a key element of the reduction, and it needs to be adjusted to the setting of $k$ matroids.

\begin{lemma}\label{lem:optEBound}
\begin{equation*}
\E[w(\OPT_w(E))] \geq \frac{p}{pk+1} \cdot \E[w(G)]\enspace.
\end{equation*}
\end{lemma}
\begin{proof}

Similar to the way we proceeded in Section~\ref{sec:makeOptsOverlap}, we define Algorithm~\ref{alg:onlineExt}, which is an extended version of Algorithm~\ref{alg:online}.
More precisely, Algorithm~\ref{alg:onlineExt} produces sets $E$ and $G$ with the same joint distribution as the ones constructed by Algorithm~\ref{alg:online}. Moreover, Algorithm~\ref{alg:onlineExt} also constructs a set $W\subseteq E$ satisfying $W\in \mathcal{F}$, and which we use to show that $\E[w(\OPT_w(E))]$ is large through the inequality $w(\OPT_w(E)) \geq w(W)$. The main difference compared to Algorithm~\ref{alg:simpExt}, which we used in Section~\ref{sec:makeOptsOverlap} to prove a related result, is that the sets $H_j$ are only updated when we add an element to either $G$ or $W$.
The proof now closely follows the reasoning of the proof of Lemma IV.8 in~\cite{feldman_2015_submodular}.

\medskip
   \begin{algorithm2e}[H]
        \DontPrintSemicolon
        \SetKwFor{Prob}{With probability}{}{endprob}
        \SetKwFor{Otherwise}{Otherwise}{}{endother}
        \SetKwFor{OthProb}{Otherwise with probability}{}{endother}
        \SetKw{stop}{stop}
        $G= \varnothing$, $W = \varnothing$, $E=\varnothing$, and $H_j=\varnothing$ \quad $\forall j\in [k]$\\
        \For{$i=1,\ldots, n$}{
          Let $e_i \in \argmax\{f_G(e) \mid e\in N\setminus \{e_1,\ldots e_{i-1}\}\}$.\\
          Let $w(e_i) = f_G(e_i)$.\\
          \If{$G+e_i \in \mathcal{F}$ and $w(e_i) > 0$}{
            \Prob{$0.5$:}{
            $G= G+e_i$.\\
             \For{$j=1,\dots, k$}{
             \lIf{$H_j + e_i \in \mathcal{I}_j$}{$H_j = H_j + e_i$.}
             \If{$H_j + e_i \not \in \mathcal{I}_j$} {$H_j = H_j - f + e_i$, where $f\in H_j\setminus G$ is any element with $H_j -f +e_i\in \mathcal{I}_j$. \label{algstep:exchange}}
              }
            }
            \OthProb{$p$ (hence, with overall probability $0.5 p$):}{
              $E=E+e_i$.\\
            \If{$H_j + e_i \in \mathcal{I}_j$ for all $j\in [k]$}{
              $W = W + e_i$.\\
              $H_j = H_j +e_i$ for $j\in [k]$.
            }
            }

          }
        }

    \caption{$\onlineExt(p)$}\label{alg:onlineExt}
   \end{algorithm2e}
\medskip

Consider an arbitrary element $e\in N$ processed by Algorithm~\ref{alg:onlineExt}, and fix all history up to that point. Let $G'$, $W'$, and $H_j'$ for $j\in [k]$ be the sets $G$, $W$, and $H_j$ at that moment. Notice that, since we fixed the history, the weight $w(e)$ of $e$ is no longer a random variable.

We are interested in studying the expected increase in $w(G')$ and $w(W')$ following the processing of $e$. If $G'+e\not\in \mathcal{F}$ or $w(e) \leq 0$, then the expected increase in the weight of both sets is $0$. Now, consider the case where $G'+e\in \mathcal{F}$ and $w(e) > 0$. The expected increase in $w(G')$ is equal to $w(e)/2$ because $e$ gets added to $G'$ with probability $0.5$.
The expected increase in $w(W')$ depends on whether $H_j + e \in \mathcal{I}_j$ for all $j\in [k]$. If this is the case, then the expected increase in $w(W')$ is equal to $(p/2)\cdot w(u)$ because $e$ gets added to $W'$ with probability $p/2$. Otherwise, if there is a $j\in [k]$ such that $H_j +e \not\in \mathcal{I}_j$, then the expected increase in $w(W')$ is $0$. To compensate for this, we charge in this case $(p/2)\cdot w(u)$ units to the element $f\in H'_j\setminus G'$ that is removed from $H'_j$ in line~\ref{algstep:exchange} if $e$ gets added to $G$. If several elements from different sets $H'_{j'}$ (for different $j'\in [k]$) get removed, we only charge one arbitrarily chosen element among them. It is important to note that we do this charging regardless of whether $u$ actually ends up being added to $G'$ or not.
In summary, if $e$ is such that $G'+e\in \mathcal{F}$ and $w(e) >0$, then the expected increase of $w(G')$ is $w(u)/2$ and we either have an expected increase of $(p/2)\cdot w(u)$ for $w(W')$ or increase the load of an element by $(p/2)\cdot w(u)$. Let $c(f)$ be the total load that an element $f\in N$ obtains through this process. Notice that $c(f)$ is a random variable. By the above explanation we have
\begin{equation}\label{eq:chargeRel}
\frac{\E[w(G)]}{\E[w(W)] + \sum_{e\in N}\E[c(e)]} = \frac{1/2}{p/2} = \frac{1}{p}\enspace.
\end{equation}

We now analyze how large the expectation of all charged values is. For this observe that only elements in $W$ get charged, because we only charge elements $f$ that may be removed from a set $H_j$ in line~\ref{algstep:exchange} of Algorithm~\ref{alg:onlineExt}; hence, these are elements that are not part of $G$ but are contained in $H_j$ at that moment of the algorithm. Since $H_j$ contains a subset of elements of $G$ and $W$, the element $f$ must be contained in $W$.
Now, consider an element $f\in W$. Whenever Algorithm~\ref{alg:onlineExt} considers an element $e\in N$ that, if added to $G$, would lead to $f$ being removed from at least one set $H_j$ for $j\in [k]$, then $f$ has a probability of $0.5$ to be removed from $H_j$ (the probability of $e$ to be added to $G$). Thus, for a fixed $j$, the number of times that $f$ will be charged due to $H_j$ is upper bounded by a random variable $Y$ with geometric distribution of parameter $0.5$. Moreover, notice that $w(e)\leq w(f)$, whenever we consider an element $e$ that leads to $f$ being charged, which follows by the fact that the weights of considered elements are non-increasing in Algorithm~\ref{alg:onlineExt}. Hence, the expected load of element $f$ due to a conflict in $H_j$ is upper bounded by $\E[Y] \cdot (p/2) w(f) = p\cdot w(f)$. Since there are $k$ sets $H_j$ due to which $f$ can be charged, we obtain the following overall bound on the expected load on an arbitrary element $f\in N$:
\begin{equation*}
\E[c(f)] \leq k\cdot p \cdot \E[w(W\cap \{f\})]\enspace.
\end{equation*}
Using the above bound in~\eqref{eq:chargeRel}, we finally obtain
\begin{align*}
\frac{\E[w(G)]}{\E[w(W)] + kp\E[w(W)]} \leq \frac{1}{p}
\quad \implies \quad
\frac{p}{pk+1}\cdot \E[w(G)] \leq \E[w(W)]\enspace,
\end{align*}
and the result follows from $w(\OPT_w(E)) \geq w(W)$.
\end{proof}

Combining Observation~\ref{obs:linRatio} and Lemma~\ref{lem:optEBound} we obtain the following.
\begin{corollary}\label{cor:linValueLargeWrtG}
\begin{equation*}
\E[w(Q\cap E)] \geq \frac{\alpha p}{pk+1}\cdot \E[w(G)]\enspace.
\end{equation*}
\end{corollary}

The following proposition corresponds to Proposition IV.10 in~\cite{feldman_2015_submodular} with the only difference that we obtain a different factor on the right-hand side in front of $f(\OPT)$ due to the fact that its proof is based on Corollary~\ref{cor:gLargeWrtOpt} and Corollary~\ref{cor:linValueLargeWrtG} which we use for general $k$ instead of $k=1$. Otherwise, the proof is identical, and we therefore refer the interested reader to~\cite{feldman_2015_submodular} for further details.

\begin{proposition}[See Proposition IV.10 in~\cite{feldman_2015_submodular}]
\label{prop:globalGuaranteeDependingOnLinSubDisc}
If $\E[w(Q\cap E) - f_w(Q\cap E)] \leq q\cdot \E[w(E) - f_w(E)]$ for some value $q\geq \alpha$, then
\begin{equation*}
\E[f(Q\cap E)] \geq \E[f_w(Q\cap E)] \geq \frac{p\left(\alpha+\alpha p-2qp (pk+1)\right)}{4(k+1)(pk+1)(1+p)} \cdot f(\OPT)\enspace.
\end{equation*}
\end{proposition}

The following observation, which is Observation~IV.11 from~\cite{feldman_2015_submodular}, shows that the above proposition can always be applied with $q=1$.

\begin{observation}[Observation~IV.11 in~\cite{feldman_2015_submodular}]\label{obs:diffWToFWSmall}
\begin{equation*}
w(Q\cap E) - f_w(Q\cap E) \leq w(E) - f_w(E)\enspace.
\end{equation*}
\end{observation}

Finally, we obtain Theorem~\ref{thm:mainSubmToLin}, which implies our main result (Theorem~\ref{thm:mainSubmToLin}) as discussed above.

\begin{reptheorem}{thm:mainSubmToLin}
By choosing $p=\sfrac{\alpha}{3k}$ in Algorithm~\ref{alg:online}, we obtain
\begin{equation*}
\E[f(Q\cap E)] \geq \frac{\alpha^2}{128 k^2}\cdot f(\OPT) \enspace.
\end{equation*}
\end{reptheorem}
\begin{proof}
By Observation~\ref{obs:diffWToFWSmall}, we can choose $q=1$ in Proposition~\ref{prop:globalGuaranteeDependingOnLinSubDisc}. The theorem now follows by plugging in $q=1$ and $p=\sfrac{\alpha}{3 k}$ in the inequality stated in Proposition~\ref{prop:globalGuaranteeDependingOnLinSubDisc}, and simplifying the resulting expression by using $\alpha \leq 1$ and $k\geq 1$, which leads to
\begin{align*}
\E[f(Q\cap E)]
	&\geq
	\frac{\alpha\left(9k\alpha+3\alpha^2 - 2\alpha (\alpha+3)\right)}{12k(k+1)(\alpha+3)(3k+\alpha)} \cdot f(\OPT)\\
	&=
\frac{9 \alpha^2 k - 6 \alpha^2 + \alpha^3}{12k {\left[3 {\left(3  + \alpha\right)} k^2 + {\left(9  + 6 \alpha + \alpha^2\right)} k + {\left(3 \alpha + \alpha^2 \right)}\right]}}\cdot f(\OPT)\\
 &\geq \frac{ 3 \alpha^2 k}{12k {\left[3 {\left(3  + 1\right)} k^2 + {\left(9  + 6  + 1\right)} k^2 + {\left(3 + 1 \right)} k^2\right]}}\cdot f(\OPT)
=
\frac{\alpha^2}{128 k^2} \cdot f(\OPT)\enspace.
\qedhere
\end{align*}
\end{proof}

\end{document}